\documentclass[11pt]{article}

\usepackage[a4paper, margin=1.25in]{geometry}
\usepackage{fourier} 
\usepackage{booktabs} 
\usepackage[T1]{fontenc} 
\usepackage{authblk} 
\usepackage{footmisc}

\usepackage[dvipsnames]{xcolor} 
\definecolor{orcidgreen}{HTML}{85A12C} 
\usepackage{url}
\usepackage[hidelinks]{hyperref} 
\hypersetup{
    colorlinks=true,
    allcolors=orcidgreen,
    linkcolor=Blue,
    citecolor=purple,
}
\usepackage{microtype} 

\usepackage{amssymb}
\usepackage{amsmath}
\usepackage{amsthm}
\usepackage{orcidlink}
\usepackage{academicons}
\usepackage{xspace}
\usepackage{nicefrac}

\usepackage[nameinlink]{cleveref} 

\newcommand{\reals}{\mathbb{R}}
\newcommand{\pref}{\succ}

\newcommand{\swap}{{\mathrm{swap}}}
\newcommand{\ham}{{\mathrm{ham}}}

\newcommand{\diam}{{\mathrm{diam}}}
\newcommand{\calL}{{\mathcal{L}}}
\newcommand{\calD}{{\mathcal{D}}}
\newcommand{\calT}{{\mathcal{T}}}
\newcommand{\kkemscore}{{{k\hbox{-}\mathrm{Kemeny}}}}
\newcommand{\kemscore}{{{\mathrm{Kemeny}}}}

\newcommand{\kemeny}{{{\mathrm{Kemeny}}}}

\newcommand{\np}{{{\mathrm{NP}}}}
\newcommand{\fpt}{{{\mathrm{FPT}}}}

\newcommand{\probName}[1]{\textsc{#1}\xspace}
\newcommand{\onekemeny}{\probName{$1$-Kemeny}}
\newcommand{\kkemeny}{\probName{$k$-Kemeny}}

\usepackage{tikz}
\usetikzlibrary{calc}
\usepackage{subcaption}

\usepackage[square]{natbib}
\usepackage{cleveref}

\theoremstyle{plain}
\newtheorem{theorem}{Theorem}[section]
\newtheorem{lemma}[theorem]{Lemma}
\newtheorem{proposition}[theorem]{Proposition}
\newtheorem{corollary}[theorem]{Corollary}

\newtheorem{definition}{Definition}[section]

\newtheorem{remark}{Remark}[section]

\usepackage{todonotes}

\usepackage{etoolbox}

\newcommand{\appendixProofs}{}

\newcommand{\appendixproof}[3]{%
  \gappto{\appendixProofs}{
		\subsection{Proof of \Cref{#2}}\label{proof:#2}
    #3
  }
}

\title{\huge \textbf{Diversity of Structured Domains\\
via $\boldsymbol{k}$-Kemeny Scores}}

\author{
  Piotr Faliszewski\quad\quad
  Krzysztof Sornat\quad\quad
  Stanisław Szufa\quad\quad
  Tomasz Wąs
}

\date{}

\makeatletter

\renewcommand*{\@fnsymbol}[1]{\ifcase#1\or i\or ii\or iii\or iv\else\@ctrerr\fi}
\makeatother

\begin{document}

\maketitle

\begin{quote}
    \textbf{Abstract:}
    In the $k$-\textsc{Kemeny} problem, we are given an ordinal
    election, i.e., a collection of votes ranking the candidates from
    best to worst, and we seek the smallest number of swaps of adjacent
    candidates that ensure that the election has at most $k$
    different rankings. We study this problem for a number of structured
    domains, including the single-peaked, single-crossing,
    group-separable, and Euclidean ones. We obtain two kinds of results:
    (1) We show that $k$-\textsc{Kemeny} remains intractable under most
    of these domains, even for $k=2$, and (2) we use $k$-\textsc{Kemeny}
    to rank these domains in terms of their diversity.
\end{quote}


\section{Introduction}\label{sec:intro}

\renewcommand{\thefootnote}{} 
\footnotetext{
  \hspace{-19pt}
  Authors' Information:
  \href{https://orcid.org/0000-0002-0332-4364}{Piotr Faliszewski \orcidlink{0000-0002-0332-4364}}, faliszew@agh.edu.pl, AGH University, Poland;
  \href{https://orcid.org/0000-0001-7450-4269}{Krzysztof~Sornat~\orcidlink{0000-0001-7450-4269}}, sornat@agh.edu.pl, AGH University, Poland;
  \href{https://orcid.org/0000-0001-6301-6227}{Stanisław Szufa \orcidlink{0000-0001-6301-6227}}, s.szufa@gmail.com, CNRS, Université Paris Dauphine-PSL, France;
  \href{https://orcid.org/0000-0003-3492-6584}{Tomasz Wąs \orcidlink{0000-0003-3492-6584}}, tomasz.was@cs.ox.ac.uk, University of Oxford, United Kingdom.
}
\renewcommand{\thefootnote}{\arabic{footnote}}

An \emph{ordinal election} consists of a set of candidates and a
collection of votes, ranking these candidates from the most to the
least desirable one, where each vote comes from a given
\emph{domain}.
We study the diversity of structured domains, such as the
single-peaked~\citep{bla:b:polsci:committees-elections},
single-crossing~\citep{mir:j:single-crossing,rob:j:tax},
group-separable~\citep{ina:j:group-separable,ina:j:simple-majority},
and Euclidean ones~\citep{ene-hin:b:spatial,ene-hin:b:spatial2}, as
well as the diversity of elections with votes from these domains (see
\Cref{sec:prelim} for detailed definitions). In essence, structured
domains restrict possible votes to those that are somehow
\emph{reasonable}; for example, in the single-peaked domain over the
standard political left-right axis, one could not rank the extreme
left-wing and right-wing candidates on the two top positions.
To capture diversity, we employ a technique based on solving the
$k$-Kemeny problem of \citet{fal-kac-sor-szu-was:c:div-agr-pol-map}.
The idea is that an election---or, a structured
domain---is diverse if it includes many very different votes that
cannot be easily grouped into (a small number of) clusters.  
Our results come in two main flavors.
First, we establish the computational complexity of the $k$-Kemeny 
problem across our domains. 
Second, we rank these domains---as well as several 
statistical cultures used to sample elections 
from them---with respect to their diversity.

Studying the diversity of structured domains 
and structured elections is
important for several reasons:
\begin{enumerate}
\item Diversity is a fundamental property of elections and domains,
  yet formally capturing this intuitive notion is challenging 
  (see discussion of the related work below). Hence, it is
  valuable to both develop tools for analyzing diversity and to
  use these tools to gain better insights into various 
  domains and elections.
\item \citet{fal-kac-sor-szu-was:c:div-agr-pol-map} argued that one
  can understand the nature of elections---specifically, their
  locations on the map of
  elections~\citep{szu-boe-bre-fal-nie-sko-sli-tal:j:map,fal-kac-sor-szu-was:c:div-agr-pol-map,boe-fal-nie-szu-was:c:distance-measures-zoometrics}---by
  analyzing their three natural properties, including diversity (the
  other two are polarization and agreement).
\item When designing numerical experiments on elections, one may wish
  to consider synthetic datasets with different levels of
  diversity. In particular, one may choose various statistical
  cultures---i.e., probabilistic distributions over votes---based on
  the diversity of the data they produce.
  \end{enumerate}

The classic Kemeny score of an election is the smallest number of
swaps of adjacent candidates required to ensure that all votes are
identical (in \Cref{sec:prelim} we give a different but equivalent
definition; this one follows the distance-rationalization
framework~\citep{bai:j:distance-rationalisation,mes-nur:b:distance-realizability,elk-fal-sli:j:dr}).
Similarly, the $k$-Kemeny score is the smallest number of swaps, which
ensure that the election consists of at most $k$ different
votes. Intuitively, to compute the $k$-Kemeny score we need to find
$k$ groups of similar votes (thereby solving a clustering problem) and
compute their Kemeny scores separately.
\citet{fal-kac-sor-szu-was:c:div-agr-pol-map} argued that a weighted
sum of $k$-Kemeny scores for different values of $k$ gives a good
measure of diversity; the same approach, albeit with different
weights, was taken by
\citet{fal-mer-nun-szu-was:c:map-dap-top-truncated}.

\subsection{Our Contributions}
We obtain the following sets of results. First, we find a
polynomial-time algorithm for computing the $k$-Kemeny score of
single-crossing elections, but show $\np$-hardness for single-peaked
and group-separable ones, already for $k=2$. For Euclidean elections,
the results are more varied and depend on the exact
assumptions.

Second, using $k$-Kemeny scores, we rank our structured
domains from the most to the least diverse one. Surprisingly, the
caterpillar group-separable domain turns out to be the most diverse
one, suggesting that it might deserve to be used in experiments more
often.

Third, we find that the typical way of sampling Euclidean elections,
used in most computational social choice papers that employ these
models---see the survey of
\citet{boe-fal-jan-kac-lis-pie-rey-sto-szu-was:c:guide}---by necessity
does not generate a sizable fraction of votes that belong to this
domain. We show how this affects the diversity of generated elections,
as compared to sampling votes from these domains uniformly at random.

We also make a number of remarks about our domains, showing
their various quirks. All proofs are in the appendix.

\subsection{Related Work}
The two most related papers are those of
\citet{fal-kac-sor-szu-was:c:div-agr-pol-map,fal-mer-nun-szu-was:c:map-dap-top-truncated},
where the authors introduce and use measures of diversity based on
computing weighted sums of $k$-Kemeny scores. The former work also
argues that many previously studied notions of diversity---such as
those considered by
\citet{alc:vor:j:cohesiveness,alc:vor:j:cohesiveness2},
\citet{can-ozk-sto:j:polarization,can-ozk-sto:j:polarization2} and
\citet{has-end:c:diversity-indices}---capture (dis)agreement among the
votes rather than their diversity; in particular, they conflate the
notions of diversity and polarization.

Recently, \citet{amm-pup:j:domain-diversity} developed a number of
diversity notions for preference domains, based on the multi-attribute
approach of \citet{neh-pup:j:diversity}. The main idea is that the
diversity of a domain depends on how many (possibly weighted)
``attributes'' its votes have, where an attribute can be a property
such as ``candidate $c$ is ranked on top.'' However,
\citet{amm-pup:j:domain-diversity} mostly focus on the properties of
their measures and not on applying them to particular domains.
\citet{kar-mar-rii-zho:t:domain-diversity} analyze how many different
rankings of at least $s$ candidates appear in a given Condorcet
domain. While we focus on several mainstream domains, they consider
numerous special ones. Overall, the main difference between our
approach and those of \citet{amm-pup:j:domain-diversity} and
\citet{kar-mar-rii-zho:t:domain-diversity} is that they count
occurrences of particular structures in the votes, whereas our $k$-Kemeny-based approach analyzes interrelations between
these votes.

The problem of computing the Kemeny score was shown to be $\np$-hard
by \citet{bar-tov-tri:j:who-won}, and its exact complexity was
established by \citet{hem-spa-vog:j:kemeny}. There are various ways of
circumventing these intractability results, ranging from approximation
algorithms (see, e.g., the work
of~\citet{ail-cha-new:j:kemeny-approx}), through parameterized
approaches (see, e.g., the works of
\citet{bet-fel-guo-nie-ros:j:fpt-kemeny-aaim,bet-bre-nie:j:kemeny}),
and heuristics (see, e.g., the work of \citet{con-dav-kal:c:kemeny}).
While it is well-known that Kemeny score can be computed in
polynomial-time for Condorcet
domains~\citep{bar:j:kemeny-condorcet,tru:t:kemeny-condorcet}, doing
so is also possible if the input election is, in a certain formal
sense, close to being in some such
domains~\citep{cor-gal-spa:c:sp-width,cor-gal-spa:c:spsc-width}.  Researchers also considered
the complexity of computing Kemeny scores in Euclidean
elections~\citep{esc-spa-tyd:j:kemeny-2d,ham-lac-rap:c:kemeny-2d-embedding}.
Computational aspects of $k$-Kemeny scores were, so far, considered
only by \citet{fal-kac-sor-szu-was:c:div-agr-pol-map}.

\section{Preliminaries}\label{sec:prelim}

Given a set of candidates $C = \{c_1, \ldots, c_m\}$, we write
$\calL(C)$ to denote the set of all strict rankings (linear orders)
over $C$, and we refer to $\calL(C)$ as the \emph{full domain} over
$C$. We often focus on various other domains $\calD$ that are subsets of
$\calL(C)$, typically referred to as \emph{structured domains}. For a
ranking $v$ and candidates $a$ and $b$, we write $v \colon a \pref b$
to indicate that $v$ ranks $a$ higher than $b$. If $A$ and $B$ are two
disjoint subsets of candidates, then by $v \colon A \pref B$ we mean
that $v$ prefers each member of $A$ to each member of $B$.  The
\emph{swap distance} of two rankings $u, v \in \calL(C)$ is the number
of swaps of adjacent candidates needed to transform $u$ into $v$
(equal to the number of inversions, i.e., the number of pairs of
candidates that are ranked differently in $u$ and $v$).  

We assume
knowledge of standard notions from computational complexity
theory~\citep{pap:b:complexity} and its parameterized
variant~\citep{cyg-fom-kow-lok-mar-pil-pil-sau:b:parameterized-algorithms}.

\subsection{Elections}
An \emph{election} $E = (C,V)$ consists of a set $C$ of candidates and
a collection $V$ of voters, where every voter has a vote from
$\calL(C)$, ranking the candidates from the most to the least
desirable one. To streamline our discussion, we use the same symbols
to refer to both the voters and their votes, with the exact meaning
clear from the context.  We often focus on elections where the votes
are restricted to belong to some structured domain, rather than the
full one.

\paragraph[(k-)Kemeny Rankings]{($\boldsymbol{k}$-)Kemeny Rankings.}
Let us fix an election $E = (C,V)$ and let $r$ be some ranking from
$\calL(C)$. Its \emph{Kemeny score} is:
\[
  \kemeny_E(r) = \textstyle\sum_{v \in V} \swap(v,r).
\]
A ranking with the lowest Kemeny score is known as a \emph{Kemeny ranking}~\cite{kem:j:no-numbers}.
Similarly, the $k$-Kemeny score of a set $R = \{r_1, \ldots, r_k\}$ of
$k$ rankings is:
\[
  k\hbox{-}\kemeny_E(R) = \textstyle\sum_{v \in V} \min_{i\in[k]}\swap(v,r_i).
\]
We refer to the set that minimizes the $k$-Kemeny score as the
\emph{$k$-Kemeny set} and to its members as \emph{$k$-Kemeny
  rankings}~\citep{fal-kac-sor-szu-was:c:div-agr-pol-map}.  We study
the problem below.

\begin{definition}
  In the \textsc{$k$-Kemeny} problem we are given an election $E$ as
  well as integers $k$ and $q$, and we ask if there is a set $R$ of
  $k$ rankings such that $k\hbox{-}\kemeny_E(R) \leq q$.
\end{definition}

\subsection{Structured Domains}
In addition to the full domain, we also consider its various
structured variants defined below:
\begin{description}
\item[Single-Peaked Domain (SP).] Let $\lhd$ be some ranking from
  $\calL(C)$, referred to as an axis. The \emph{single-peaked domain (SP)} for
  $\lhd$ consists of all rankings $v \in \calL(C)$ that satisfy the
  following condition: For every $t \in [|C|]$ the top $t$ candidates
  in $v$ form an interval within $\lhd$. This domain was introduced
  by \citet{bla:b:polsci:committees-elections}.

\item[Single-Crossing Domains (SC).] A subset of $\calL(C)$ is
  \emph{single-crossing (SC)} if it is possible to order its members as
  $(v_1, \ldots, v_n)$, so that for each pair of candidates
  $a, b \in C$ there is a number $t_{ab} \in [n]$ such that voters
  $v_1, \ldots, v_{t_{ab}}$ rank $a$ and $b$ in one way, and the
  remaining ones rank them in the opposite way. These domains were
  introduced by \citet{mir:j:single-crossing} and \citet{rob:j:tax}.

\item[Group-Separable Domains (GS).] Let $\calT$ be a rooted, ordered
  tree, where each leaf is labeled with a unique candidate from $C$
  and each internal node has at least two children. A vote
  $v \in \calL(C)$ is \emph{consistent} with $\calT$ if we can obtain
  it by reading the labels of the leaves from left to right, after
  possibly reversing the order of some nodes' children. A
  \emph{group-separable domain (GS)} for tree $\calT$ contains all
  votes consistent with $\calT$. The notion of group-separability is
  due to \citet{ina:j:group-separable,ina:j:simple-majority}, but we
  follow the equivalent definition of \citet{kar:j:group-separable}.
  Whenever we speak of GS, we mean its variant for a given binary
  tree.
\end{description}

It is well-known that all single-peaked domains for candidate sets of
the same cardinality are isomorphic, so we typically speak of
\emph{the} single-peaked domain. On the other hand, if there are at
least three candidates then there are many different single-crossing
domains; these observations are made explicitly, e.g., by
\citet{fal-sko-sli-szu-tal:c:isomorphism-jcss}. Group-separable
domains are isomorphic, provided that the underlying trees are
isomorphic.  We are particularly interested~in:
\begin{description}
\item[Balanced Group-Separable Domain (GS/bal).]  A~GS
  domain is \emph{balanced} if its underlying tree is a binary tree
  where each level---except possibly the last one---is completely
  filled (we refer to such trees as \emph{balanced})

\item[Caterpillar Group-Separable Domain (GS/cat).]  A GS domain is
  \emph{caterpillar} if its underlying tree is binary, where each
  internal node has at least one leaf as a child (we refer to such
  trees as \emph{caterpillar}).
\end{description}

A domain $\calD$ is a \emph{Condorcet domain} if for each election
in $\calD$ there is a ranking $r$, called \emph{Condorcet ranking}, such that for each two
candidates $a$ and $b$, $r \colon a \pref b$ implies that at least half of the
voters prefer $a$ to $b$. All the above domains are Condorcet.
Other domains we consider are:

\begin{description}
\item[Single-Peaked on a Graph Domains (SP/G).]  For a graph $G$ where
  each vertex is labeled with a unique candidate, a vote
  $v \in \calL(C)$ is single-peaked with respect to $G$ if for each
  $t \in [|C|]$ the graph induced by the top $t$ candidates is
  connected (so, for a path we get the classic single-peaked domain).
      Similarly to GS, in each SP/G domain we assume a specific connected graph $G$
  for each candidate set, which we can compute in polynomial-time.
  This domain was
  mentioned, e.g., by \citet{elk-lac-pet:b:structured-domains}; its
  variant for trees is due to \citet{dem:j:sp-trees}.
  
\item[Single-Peaked on a Circle Domain (SPOC).] This is the
  SP/G domain for the case where graph $G$ is a cycle; it is due to \citet{pet-lac:j:spoc}.
  
\item[$\boldsymbol{d}$-Euclidean Domains.] Let
  $d$ be a positive integer and let $x \colon C \rightarrow \reals^d$
  be a function that associates each candidate with a point in the
  Euclidean space (called \emph{embedding function}). A vote
  $v \in \calL(C)$ belongs to the domain induced by $x$
  if there is a point $x_v \in \reals^d$ such that for
  each two candidates $a, b \in C$, if $v \colon a \pref b$ then
  the Euclidean distance between $x_v$ and $x(a)$ is smaller than
  between $x_v$ and $x(b)$.
  These domains are discussed, e.g., by
  \citet{ene-hin:b:spatial,ene-hin:b:spatial2}.
\end{description}

\paragraph{Statistical Cultures.} Fix a candidate set $C$ and some
domain $\calD \subseteq \calL(C)$. A statistical culture is a
distribution over the votes from $\calD$. In particular, by
\emph{impartial culture over $\calD$} we mean the uniform distribution
over $\calD$. If we omit $\calD$, then we mean impartial culture of
$\calL(C)$. We introduce further statistical cultures in
\Cref{sec:setup-experiments}.

\section[Complexity of k-Kemeny]{Computational Complexity of $\boldsymbol{k}$-\textsc{Kemeny}}

The \kkemeny problem is $\np$-hard even if $k=1$ and $n=4$
voters~\citep{dwo-kum-nao-siv:c:rank-aggregation,bie-bra-den:j:kemeny-hardness}. However, for
Condorcet domains and $k=1$ it can be solved trivially
(it suffices to compute the Kemeny score of the Condorcet ranking,
which is guaranteed to be optimal).
We show that for some prominent Condorcet domains, including SP,
GS/bal, and GS/cat, as well as for SP/G domains, \kkemeny becomes
$\np$-hard already for $k=2$.  We also discuss the complexity of
\kkemeny on Euclidean domains. We supplement these results with a
general $\fpt$ algorithm for Condorcet domains, parameterized by the
number of voters, and an outright polynomial-time algorithm for SC
elections.

\subsection{Intractability Results for Structured Domains}

The key idea of our hardness proofs is to give reductions from the
\textsc{Hypercube 2-Segmentation} (\textsc{H2S}) problem. To define
this problem, we need some additional notation. For a binary string
$x$, we write $x[j]$ to refer to its $j$-th symbol. \emph{Hamming}
distance between two equal-length strings $x$ and $y$, denoted
$\ham(x,y)$, is the number of positions on which these two strings
differ.  For a sequence $S = (s_1, \ldots, s_n)$ of binary strings,
each of length $m$, the Hamming distance between $S$ and another
binary string $r$ of length $m$ is
$\ham(S,r) = \sum_{i=1}^n \ham(s_i,r)$. A string that minimizes the
Hamming distance to $S$ is called \emph{central} for $S$ and for every
position $i \in [m]$, has the same symbol on this position as at least
half of the strings in $S$ (in case of a tie, a central string can
take either symbol). The Hamming distance of such a string to $S$ is
the \emph{Hamming score} of $S$, denoted $\ham(S)$.

\begin{definition}
  An instance of \textsc{H2S} consists of a sequence
  $S = (s_1, \ldots, s_n)$ of binary strings and of an integer $t$.
  We ask if it is possible to partition $S$ into two groups, such that
  the sum of their Hamming scores is at most $t$.
\end{definition}

The first $\np$-completeness claim for \textsc{H2S} appears in the
conference paper of \citet{kle-pap-rag:c:segmentation},
but without a proof.
Its journal version 
does not include a proof
either~\citep{kle-pap-rag:j:segmentation}, but claims
$\mathrm{MAXSNP}$-hardness (also without proof). The $\np$-hardness
proof was eventually presented 16
years later  
by \citet{fei:t:segmentation}, who also argued why the
$\mathrm{MAXSNP}$-hardness claim
was incorrect.

\begin{theorem}\label{thm:hard:sp:gsbal}
  \kkemeny is $\np$-complete even for $k=2$ and elections that are
  both SP and GS/bal.
\end{theorem}
\begin{proof}[Proof sketch]
  To show $\np$-hardness, we reduce from \textsc{H2S}.
  Let the input instance consist of
  an integer $t$ and a sequence $S = (s_1, \ldots, s_n)$ of binary
  strings, each of length $m$.
  We form a
  set of candidates
  $C = \{a_1, \ldots, a_m\} \cup \{b_1, \ldots, b_m\}$ and we say that
  a preference order $v$ is \emph{aligned} if it is of the following
  form:
  \[
    v \colon \{a_1,b_1\} \pref \{a_2,b_2\} \pref \cdots \pref \{a_m,b_m\}.
  \]
  An aligned vote is consistent with length-$m$ binary string $x$ if
  for each $j \in [m]$, we have $v \colon a_j \pref b_j$ when $x[j]=1$
  and $v_i \colon b_j \pref a_j$ if $x[j]=0$.  We form an election
  $E = (C,V)$, where for each string $s_i$ we have exactly one vote
  $v_i$, aligned and consistent with it. We form a \kkemeny instance
  with election $E$, $k=2$, and $q = t$.  We observe that election
  $E$ is single-peaked with respect to societal axis:
  \[
    a_m \lhd \cdots \lhd a_2 \lhd a_1 \lhd b_1 \lhd b_2 \lhd \cdots
    \lhd b_m.
  \]
  It is also balanced group-separable, as witnesses by balanced binary
  tree $\calT$ with $2m$ leaves, where reading the labels of the
  leaves from left to right gives order
  $a_1 \pref b_1 \pref a_2 \pref b_2 \pref \cdots \pref a_m \pref
  b_m$. We ask if there are two rankings $r'$ and $r''$ such that $\kkemscore_E(\{r',r''\}) \leq t$.
\end{proof}
\appendixproof{Theorem}{thm:hard:sp:gsbal}{
  It is immediate to see that \kkemeny is in $\np$. To show
  $\np$-hardness, we give a reduction from \textsc{H2S}.  Let
  the input instance consist of an integer $t$ and a sequence
  $S = (s_1, \ldots, s_n)$ of binary strings, each of length $m$,
  where $m$ is a power of two (this assumption is w.l.o.g. as we can
  always ensure that it holds by extending each of the strings with no
  more than $m-1$ identical symbols). We form a set of candidates
  $C = \{a_1, \ldots, a_m\} \cup \{b_1, \ldots, b_m\}$ and we say that
  a preference order $v$ is \emph{aligned} if it is of the following
  form:
  \[
    v \colon \{a_1,b_1\} \pref \{a_2,b_2\} \pref \cdots \pref \{a_m,b_m\}.
  \]
  Further, we say that an aligned vote is consistent with length-$m$
  binary string $x$ if for each $j \in [m]$, we have
  $v \colon a_j \pref b_j$ when $x[j]=1$ and
  $v_i \colon b_j \pref a_j$ if $x[j]=0$.  We form an election
  $E = (C,V)$, where $V$ contains exactly one vote $v_i$ for each
  string $s_i$, such that $v_i$ is aligned and consistent with
  $s_i$. We form a \kkemeny instance with elections $E$, $k=2$, and
  $q = t$.  We observe that election $E$ is single-peaked with respect
  to societal axis:
  \[
    a_m \lhd \cdots \lhd a_2 \lhd a_1 \lhd b_1 \lhd b_2 \lhd \cdots
    \lhd b_m.
  \]
  It is also balanced group-separable, as witnesses by balanced binary
  tree $\calT$ with $2m$ leaves, where reading the labels of the
  leaves from left to right gives order
  $a_1 \pref b_1 \pref a_2 \pref b_2 \pref \cdots \pref a_m \pref
  b_m$.  It is immediate that the reduction can be computed in
  polynomial time. We ask if there are two rankings $r'$ and $r''$ such that $\kkemscore_E(\{r',r''\}) \leq t$.
   
  Let us now argue that the reduction is correct. Suppose there is a
  solution to the input \textsc{H2S} problem that partitions
  $S$ into $S'$ and $S''$, such that $\ham(S') + \ham(S'') \leq
  t$. Let $s'$ and $s''$ be two arbitrary central strings for $S'$ and
  $S''$. We partition the voter collection $V$ into $V'$ and $V''$
  such that each vote $v_i$ belongs to $V'$ exactly if string $s_i$
  belongs to $S'$ and we form two rankings, $r'$ and $r''$ that both
  are aligned, $r'$ is consistent with $s'$, and $r''$ is consistent
  with $s''$. It is immediate that $\kkemscore_E(\{r',r''\}) \leq t$
  as, indeed, we have:
  \begin{align*}
    \kkemscore_E(\{r',r''\})  \leq \textstyle \sum_{v_i \in V'}\swap(v_i,r')
                             + \textstyle \sum_{v_i \in V''}\swap(v_i,r'') = t.
  \end{align*}

  The other direction proceeds analogously. Assume that there are two
  rankings $r'$ and $r''$ in $\calL(C)$ such that
  $\kkemscore_E(\{r',r''\}) \leq q = t$.  By
  \Cref{pro:kemeny-condorcet}, we know that both $r'$ and $r''$ are
  aligned, and so we can derive length-$m$ binary strings $s'$ and
  $s''$ such that $r'$ and $r''$ are consistent with $s'$ and $s''$,
  respectively. We partition $S$ into $S'$ and $S''$ so that each
  string $s_i$ belongs to $S'$ exactly if
  $\swap(v_i,r') \leq \swap(v_i,r'')$, and it belongs to $S''$
  otherwise. One can verify that $\ham(S',s') + \ham(S'',s'') \leq t$
  and, hence, $\ham(S') + \ham(S'') \leq t$.
}

Before we discuss an analogous result for caterpillar group-separable
elections, we make an observation about how one can 
generate
GS/cat votes, or verify that votes are GS/cat.
Take an axis $c_1 \lhd \cdots \lhd c_m$, corresponding to a
caterpillar tree (the 
leaves, read from left to right, are labeled with $c_1, \ldots,
c_m$). To form a GS/cat vote for this axis, we consider the candidates
in the order $c_1, c_2, \ldots, c_m$ and for each $c_i$ we choose
whether to place it on the highest or the lowest still available
position. We refer to this as \emph{caterpillar vote construction
  (CVC)}.
The appeal is that if we sample decisions in CVC uniformly at random,
it is very similar to uniform sampling of SP votes: There, we always
place the considered candidate in the lowest available position,
randomizing between selecting the top- or bottom not-yet-ranked
candidate from the axis~\citep{wal:t:generate-sp};
a similar relation between GS/cat and SP was already noted by
\citet{boe-bre-elk-fal-szu:c:map-pref-learning}.

\begin{theorem}\label{thm:hard:gscat}
  \kkemeny is $\np$-complete even for $k=2$ and elections that are caterpillar
  group-separable.
\end{theorem}
\begin{proof}[Proof sketch]
  We sketch a reduction from \textsc{H2S}, similar in spirit to the
  one for \Cref{thm:hard:sp:gsbal}, but with a different representation
  of the votes. Consider an \textsc{H2S}  instance with sequence
  $S = (s_1, \ldots, s_n)$ of binary strings of length-$m$ each, and
  an integer $t$. 
  Let $M = m^{10}n^{10}$.

  We form a candidate set $C = A \cup B \cup X$, where
  $A = \{a_1, \ldots, a_m\}$, $B = \{b_1, \ldots, b_m\}$ and
  $X = \{x_1, \ldots, x_M\}$, where $X$ is a large set of dummy
  candidates.  For each binary string $z$ of length $m$ and each
  position $j \in [m]$, we let $c_z(j)$ be $a_j$ if $z[j]=1$, and we
  let $c_z(j)$ be $b_j$ if $z[j]=0$. By $\overline{c}_z(j)$ we mean
  the unique candidate in $\{a_j,b_j\} \setminus \{c_z(j)\}$. For
  each binary string $z$ of length $m$, we define ranking $r(z)$:
  \begin{align*}
    r(z) \colon c_z(1) \pref  \cdots \pref c_z(m) \pref X  \pref \overline{c}_z(m) \pref \cdots  \pref \overline{c}_z(1),
  \end{align*}
  where by $X$ we mean listing members of $X$ from $x_1$ to $x_M$.
  For example, if $z$ has prefix $101$, then $r(z)$ is of the
  form:
  \[
    r(101...) \colon a_1 \pref b_2 \pref a_3 \pref \cdots \pref X \pref \cdots \pref b_3 \pref a_2 \pref b_1.
  \]
  We form an election $E$ with candidate set $C$ and a single vote
  $r(s_i)$ for each string $s_i \in S$. Note that $E$ is GS/cat for
  axis
  $ a_1 \lhd b_1 \lhd a_2 \lhd b_2 \lhd \cdots \lhd a_m \lhd b_m \lhd
  x_1 \lhd \cdots \lhd x_{M}$.
  
  Our $\kkemeny$ instance consists of election $E$, $k=2$, and
  $q = 2Mt + 2nm^2$.
  Intuitively, whenever some symbol differs between a string from the
  input and the central one for its group, this corresponds to, at
  most, $2M+2m^2$ swaps within a corresponding vote ($2M$ to swap the
  respective members of $A$ and $B$ between the sides of $X$, and
  the remaining ones to arrange their final positions).
\end{proof}
\appendixproof{Theorem}{thm:hard:gscat}{
  It is clear that our problem is in $\np$ and we give a reduction
  from \textsc{H2S}. Consider its instance with sequence
  $S = \{s_1, \ldots, s_n\}$ of binary strings of length-$m$ each, and
  an integer $t$. Our proof is similar in spirit to that of
  \Cref{thm:hard:sp:gsbal}, but uses a different representation of the
  strings from the \textsc{H2S} instance.

  Let $M = m^{10}n^{10}$.  We form a candidate set
  $C = A \cup B \cup X$, where $A = \{a_1, \ldots, a_m\}$,
  $B = \{b_1, \ldots, b_m\}$ and $X = \{x_1, \ldots, x_M\}$. $X$ is a
  large set of dummy candidates; while we could take much smaller $M$,
  the current choice makes our reasoning clear and easy to follow.
  For each binary string $z$ of length $m$ and each position
  $j \in [m]$, we let $c_z(j)$ be $a_j$ if $z[j]=1$, and we let
  $c_z(j)$ be $b_j$ if $z[j]=0$. By $\overline{c}_z(j)$ we mean the
  unique candidate in $\{a_j,b_j\} \setminus \{c_z(j)\}$ (so
  $\overline{c}_z(j)$ can be seen as a complement of $c_z(j)$). For
  each binary string $z$ of length $m$, we define $r(z)$ to be the
  following preference ranking:
  \begin{align*}
    r(z) \colon c_z(1) \pref c_z(2) \pref \cdots \pref c_z(m) \pref X \pref \overline{c}_z(m) \pref \cdots \pref \overline{c}_z(2) \pref \overline{c}_z(1),
  \end{align*}
  where by $X$ we mean listing members of $X$ from $x_1$ to $x_M$.
  For example, if $z$ has prefix $101$, then $r(z)$ is of the
  form:
  \[
    r(101...) \colon a_1 \pref b_2 \pref a_3 \pref \cdots \pref X \pref \cdots \pref b_3 \pref a_2 \pref b_1.
  \]
  We form an election $E = (C,V)$ with candidate set $C$ and voter
  collection $V = (v_1, \ldots, v_n)$, where for each $i \in [n]$,
  $v_i$ has preference order $r(s_i)$. Note that this election is
  caterpillar group-separable for the following axis (this follows by
  applying the CVC procedure in a straightforward way):
  \[
    a_1 \lhd b_1 \lhd a_2 \lhd b_2 \lhd \cdots \lhd a_m \lhd b_m \lhd x_1 \lhd \cdots \lhd x_{M}.
  \]
  Our $\kkemeny$ instance consists of election $E$, $k=2$, and
  $q = 2Mt + 2nm^2$. It is immediate that the reduction can be
  computed in polynomial time.

  Let us now argue that the reduction is correct. First, assume that
  there is a partition of $S$ into $S'$ and $S''$ such that
  $\ham(S') + \ham(S'') \leq t$, and let $s'$ and $s''$ be two central
  strings for $S'$ and $S''$, respectively. We partition $V$ into $V'$
  and $V''$ so that $V'$ contains voters corresponding to the members
  of $S'$, and $V''$ contains voters corresponding to members of
  $S''$. We claim that for each vote $v_i \in V'$ it holds that:
  \[
    \swap(v_i, r(s')) \leq 2M \cdot \ham(s_i,s') + 2m^2.
  \]
  Indeed, for each position $j \in [m]$ such that $s_i[j] \neq s'[j]$,
  transforming $v_i$ into $r(s')$ requires shifting $a_j$ and $b_j$
  across the members of $X$, which requires $2M$ swaps, and then
  arranging the candidates on the top and bottom $m$ positions
  requires, altogether, at most $2m^2$ swaps. Altogether, we need at
  most $2M \cdot \ham(s_i,s') + 2m^2$ swaps per voter. As the same
  analysis applies to votes in $V''$ and $r(s'')$, altogether we have:
  \begin{align*}
    \kkemscore_E(\{r(s'),r(s'')\})  \leq 2M \big(\ham(S') + \ham(S'')\big) + 2nm^2 = 2Mt + 2nm^2.
  \end{align*}
  So, if the \textsc{H2S} instance has a solution then so does our
  \kkemeny instance.

  Let us now assume that there are rankings $r'$ and
  $r'' \in \calL(C)$ such that $\kkemscore_E(\{r',r''\}) \leq q$.  We
  partition $S$ into $S'$ and $S''$ so that a string $s_i$ belongs to
  $S'$ exactly if $\swap(v_i,r') \leq \swap(v_i,r'')$ (we could also
  use a strict inequality here, or distribute strings $s_i$ for which
  $\swap(v_i,r') = \swap(v_i,r'')$ between $S'$ and $S''$ in any
  arbitrary way).  In the analysis below we focus on $r'$, $V'$, and
  $S'$, but the reasoning for $r''$, $V''$, and $S''$ is analogous.

  Due to \Cref{pro:kemeny-condorcet}, we can assume that there is some
  partition of $A \cup B$ into sets $T$ and $F$ such that $r'$ is of
  the form $r' \colon T \pref X \pref B$ and for each $j \in [m]$ we
  either have that $a_j \in T$ and $b_j \in F$, or we have that
  $a_j \in F$ and $b_j \in T$. We form a length-$m$ binary string $s'$
  so that for each $j \in [m]$, $s'[j] = 1$ if $a_j \in T$, and
  $s'[j] = 0$ if $a_j \in F$. We observe that for each vote $v_i \in V'$
  we have:
  \[
    \swap(v_i, r') \leq 2M \cdot \ham(s_i,s') + 2m^2.
  \]
  Indeed, this follows by the same reasoning as in the first part of
  the correctness proof. Consequently, we have that (we let
  $E' = (C,V')$ and $E'' = (C,V'')$):
  \begin{align*}
    \kemscore_{E'}(r') &= \textstyle \sum_{v_i \in V} \swap(v_i, r') \\
                           &\leq 2M\cdot\ham(S') + 2|V'|\cdot m^2.
  \end{align*}
  By repeating the same analysis for $r''$, $V''$, and $S''$, we see
  that:
  \begin{align*}
    \kkemscore_{E}(\{r',r''\}) &\leq \kemscore_{E'}(r')\! +\! \kemscore_{E''}(r'') \\
                               &\leq 2M\cdot\ham(S) + 2n\cdot m^2\\ &\leq q.
  \end{align*}
  However, since $q = 2Mt + 2nm^2$, this inequality means that
  $\ham(S) \leq t$. Indeed, if $\ham(S) > t$ then we would have that
  $2M\cdot\ham(S) > 2Mt + 2nm^2$, which would be a contradiction.
  This shows that if there is a solution for our \kkemeny instance,
  then there also is one for \textsc{H2S}, which completes
  the correctness proof.
}

We can also extend
\Cref{thm:hard:sp:gsbal} to show $\np$-hardness for all
SP/G domains, which, e.g., include the SPOC domain.
\begin{theorem}
  \label{thm:hard:sp-on-graph}
  For every SP/G domain, \kkemeny is $\np$-complete even for $k=2$ and elections from this domain.
\end{theorem}
\appendixproof{Theorem}{thm:hard:sp-on-graph}{
    \kkemeny is in NP,
    thus we focus on showing the hardness,
    for which we reduce from \kkemeny for single-peaked elections.
    Given a single-peaked election $E = (C,V)$
    with $m = |C|$ candidates,
    we will construct a corresponding instance of
    SP/G election with a set of $m^2$ candidates $X$.
    We consider the same constants $k$ and $q$ 
    in both instances.

    At a high level, a graph with $m^2$ vertices has to either
    contain a path of length $m$
    or $m$ vertices that are leaves of the spanning tree.
    In both cases, we can represent each ordering on $C$
    as an ordering on these $m$ vertices
    and fix the positions of the remaining vertices.
    If we do that for all votes in election $E$
    we obtain an SP/G election,
    in which the \kkemeny problem is equivalent
    to the \kkemeny problem on $E$.

    Formally,
    let us denote the societal axis of $E$
    by $c_1 \lhd c_2 \lhd \dots \lhd c_m$
    and consider arbitrary graph $G \in \mathcal{G}$,
    with the set of vertices being $X$.
    Let $T$ be an arbitrary spanning tree of $G$.
    We will consider two cases based on whether
    the diameter of $T$,
    defined as the number of vertices in the longest path, $\diam(T)$,
    is smaller than $m$ or not.
    
    \textbf{Case 1.}
    $\diam(T) \ge m$.
    In such a case, there exists a path of $m$ vertices
    $p = (x_1, x_2, \dots, x_m)$ in $T$, and thus also in $G$.
    Since $G$ is connected,
    for the remaining vertices in $X \setminus \{x_1, x_2, \dots, x_m\}$
    we can find an ordering
    $x_{m+1},x_{m+2},\dots,x_{m^2}$ such that
    for every $i \in \{m+1,\dots, m^2\}$
    vertex $x_{i}$ is connected by an edge to
    some vertex in the set $\{x_1,x_2,\dots,x_{i-1}\}$.
    Now, let $f \colon \calL(C) \rightarrow \calL(X)$
    be a function that transforms any ranking
    $r$ over the original candidates $C$ given by
    $c_{r_1} \succ c_{r_2} \succ \dots \succ c_{r_m}$
    into a ranking over $X$ that is:
    \[
        f(r)\colon x_{r_1} \succ x_{r_2} \succ \dots \succ x_{r_m}
            \succ x_{m+1} \succ \dots \succ x_{m^2}.
    \] 
    Then, we define a new election $E' = (X,V')$
    which we obtain by transforming all of the original votes
    in this way, i.e., $V' = \{f(v) : v \in V\}$.
    
    Observe that each vote $v \in V'$ is single-peaked on $G$.
    Indeed, for each $t \in [m]$,
    the first $t$ candidates in the ranking of $v$
    form an interval on path $p$
    since $f^{-1}(v)$ is single-peaked.
    For every larger $t$, the first $t$ candidates
    form a connected subgraph of $G$
    by the way we have chosen vertices
    $x_{m+1},\dots,x_{m^2}$.
    Thus, $E'$ is single-peaked on $G$.

    Finally, observe that in the set of rankings $R'$
    that minimize $\kkemscore_{E'}(R')$,
    every ranking will be of the form
    \[
        \{x_1,x_2,\dots,x_m\} \succ x_{m+1} \succ \dots \succ x_{m^2}
    \]
    (otherwise, the distance can be decreased by moving candidates $x_i$ for $i \ge m+1$ into their respective positions).
    Also, for every ranking set $R$ over the original set of candidates $C$
    and ranking set $R' = \{f(r) : r \in R\}$,
    we have
    $\kkemscore_{E}(R) = \kkemscore_{E'}(R')$.
    Both facts imply that there exists a set of
    $k$ rankings $R' \subseteq \calL(X)$,
    with $\kkemscore_{E'}(R') \le q$,
    if and only if,
    there exists a set of $k$ rankings $R \subseteq \calL(C)$
    such that $\kkemscore_{E}(R) \le q$.
    This concludes the analysis of this case.

    \textbf{Case 2.}
    $\diam(T) < m$.
    In this case, let us first show that necessarily
    $T$ has at least $m$ leaves.
    Assume otherwise, i.e., there are at most $m-1$ leaves.
    Let $r \in X$ denote an arbitrary vertex, which we from now call a \emph{root},
    and consider paths in $T$ from every leaf to $r$.
    Every such path is of length at most $m-1$ (since $\diam(T) < m$)
    and there are at most $m-1$ of them,
    hence, in total,
    they pass through at most $(m-1)^2$ distinct vertices.
    On the other hand,
    every vertex in a tree is on the path from some leaf to the root,
    thus this number should be at least $|X| = m^2$.
    Hence, we arrive at a contradiction.

    Therefore, we can denote some $m$ leaves in $T$
    as $x_1, x_2, \dots, x_m$.
    Consider a subgraph of $T$ induced by all of the remaining vertices
    and observe that it is connected
    (as removing leaves keeps $T$ connected).
    Hence, there exists an ordering 
    $x_{m^2}, x_{m^2 - 1}, \dots, x_{m+1}$
    of vertices in $X \setminus \{x_1, x_2,\dots, x_m\}$
    such that
    for every $i \in [m^2-1]$
    vertex $x_{m^2 -i}$ is connected to some vertex in
    $\{x_{m^2}, x_{m^2 - 1}, \dots, x_{m^2 - i+1}\}$.
    
    Next, similarly as in Case 1,
    we take function
    $f \colon \calL(C) \rightarrow \calL(X)$
    that transforms any ranking
    $r$ over the original candidates $C$ given by
    $c_{r_1} \succ c_{r_2} \succ \dots \succ c_{r_m}$
    into a ranking over $X$ that is:
    \[
        f(r)\colon x_{m^2} \succ \dots \succ x_{m+1}
            \succ 
            x_{r_1} \succ x_{r_2} \succ \dots \succ x_{r_m}.
    \]
    Again, we define the new set of voters as a set of all voters in $V$
    transformed using $f$, i.e., $V' = \{f(v) : v \in V\}$.
    Then, election $E' = (X, V')$ is single-peaked on $G$.
    Indeed, for every vote $v \in V'$
    and $t \in [m^2]$
    subgraph induced by the first $t$ candidates
    is connected, by the way we defined $x_{m^2},\dots,x_{m+1}$
    and the fact that each leaf $x_{1},\dots,x_{m}$
    is connected to some vertex in $\{x_{m^2},\dots,x_{m+1}\}$.
    
    Finally, as in Case 1, we observe that
    in the set of rankings $R'$
    that minimize $\kkemscore_{E'}(R')$
    every ranking will be of the form
    \[
        x_{m^2} \succ \dots \succ x_{m+1} \succ  \{x_1,\dots,x_m\}
    \]
    and for every ranking set $R$ over the original set of candidates $C$
    and ranking set $R' = \{f(r) : r \in R\}$,
    we have
    $\kkemscore_{E}(R) = \kkemscore_{E'}(R')$.
    Therefore, there exists a set of
    $k$ rankings $R' \subseteq \calL(X)$,
    with $\kkemscore_{E'}(R') \le q$,
    if and only if,
    there exists a set of $k$ rankings $R \subseteq \calL(C)$
    such that $\kkemscore_{E}(R) \le q$.
}

For the case of $2$-Euclidean elections (and, naturally,
higher-dimensional ones), \citet{esc-spa-tyd:j:kemeny-2d} have shown
that already deciding if there is a Kemeny ranking with a given score
is $\np$-complete, even if the embedding function is given.
One of the reasons why this $\np$-hardness is possible is that
there is no guarantee that the Kemeny ranking belongs to the given
$2$-Euclidean domain
(which stands in contrast to Condorcet domains).
Indeed, if we seek a ranking that minimizes the
Kemeny score and belongs to the domain, then
\citet{ham-lac-rap:c:kemeny-2d-embedding} gave a polynomial-time
algorithm for this problem. Briefly put, the size of each
$d$-Euclidean domain is at most $O(m^{2d})$, where $m$ is the number
of candidates, so one can use brute-force search (the approach of
\citet{ham-lac-rap:c:kemeny-2d-embedding} is faster, though).
We can also perform such a brute-force
search for $k$-Kemeny scores.
\begin{definition}
  In the \textsc{$d$-Embeddable $k$-Kemeny} problem we are given an
  election $E$ over some $d$-Euclidean domain, an embedding function
  $x$ for this domain, and an integer $q$. We ask if there is a set
  $R = \{r_1, \ldots, r_k\}$ of $k$ rankings from the domain, such
  that $\kkemscore_E(R) \leq q$.
\end{definition}

\begin{corollary}  
  For each fixed $d$ and $k$, \textsc{$d$-Embeddable $k$-Kemeny} is polynomial-time solvable.
\end{corollary}

So, as opposed to \kkemeny for SP, GS/bal and GS/cat domains,
\textsc{$d$-Embeddable $k$-Kemeny} is tractable for $k=2$.  Yet, if
$k$ is part of the input then \textsc{$d$-Embeddable $k$-Kemeny} is
$\np$-complete, even for $2$-Euclidean elections.

\begin{theorem}
  \label{thm:hard:euclidean}
  \textsc{$d$-Embeddable $k$-Kemeny} is $\np$-complete for $d \geq 2$.
\end{theorem}
\appendixproof{Theorem}{thm:hard:euclidean}{
  Clearly, the problem is in NP.
  For the hardness, we will focus only on the case of $d=2$
  as 2-Euclidean domain is a subdomain
  of every $d$-Euclidean domain,
  hence the negative results carry over.

  \begin{figure}
    \centering
    \begin{tikzpicture}
        \def\xs{0.5cm} 
        \def\ys{0.5cm} 
        \def\x{0cm} 
        \def\y{0cm} 
        \def\ls{\footnotesize} 
        
        \tikzset{
            node/.style={circle, draw, minimum size=0.3cm, inner sep=0, fill = black!05},
            edge/.style={draw, thick, sloped, -, above, font=\footnotesize},
        }

        \node[node] (11) at (\x + 1*\xs, 1*\ys + \y) {};
        \node[node] (12) at (\x + 2*\xs, 1*\ys + \y) {};
        \node[node] (21) at (\x + 1*\xs, 2*\ys + \y) {};
        \node[node] (22) at (\x + 2*\xs, 2*\ys + \y) {};
        \node[node] (23) at (\x + 3*\xs, 2*\ys + \y) {};
        \node[node] (31) at (\x + 1*\xs, 3*\ys + \y) {};
        \node[node] (33) at (\x + 3*\xs, 3*\ys + \y) {};
        
        \path[edge]
        (11) edge (12)
        (11) edge (21)
        (12) edge (22)
        (21) edge (22)
        (22) edge (23)
        (21) edge (31)
        (23) edge (33)
        ;
      
    \end{tikzpicture}
    \caption{An example instance of \textsc{GridGraphDominatingSet} problem used in the proof of \Cref{thm:hard:euclidean}.}
    \label{fig:thm:hard:euclidean:1}
\end{figure}
  
  We provide a reduction from \textsc{GridGraphDominatingSet}.
  In this problem,
  we are given a constant $s \in \mathbb{N}$ and a graph $G=(V,E)$,
  in which a set of vertices is a subset of cells
  in a $n \times n$ Cartesian grid,
  i.e., $V \subseteq [n]^2$,
  and there is an edge between two vertices, if and only if,
  their cells are neighboring, i.e.,
  $\{(x_1,y_1), (x_2,y_2)\} \in E$, if and only if,
  $|x_1 - x_2| = 1$ and $y_1=y_2$, or
  $x_1 = x_2$ and $|y_1 - y_2| =1$
  (see \Cref{fig:thm:hard:euclidean:1} for an example).
  The question is whether there exists a subset $S \subseteq V$ of size $s$, known as \emph{dominating set},
  such that every vertex in the graph
  is adjacent to a vertex in $S$ or is in $S$ itself.
  The problem is known to be NP-complete~\cite{cla-col-joh:j:grid-graphs}.
  
  For every such graph $G = (V,E)$ with $V \subseteq [n]^2$ and constant $s$ we will construct a corresponding instance of 
  \textsc{$2$-Embeddable $k$-Kemeny}.
  To this end, let us first introduce some additional notation
  specific to 2-Euclidean elections.
  We say that the \emph{divider} between pair of candidates $a$ and $b$ embedded in 2-dimensional Euclidean space,
  denoted by $d_{a,b}$,
  is a perpendicular bisection of a line segment
  connecting $a$ and $b$.
  This is also the exact set of points that are equidistant from $a$ and $b$.
  Thus, every voter located on one side of $d_{a,b}$
  prefers $a$ over $b$, and every voter on the other side
  prefers $b$ over $a$.
  Furthermore, for any two voters $u$ and $v$
  that do not lay any divider
  (and we will consider only such voters),
  the number of dividers we cross
  when we go from $u$ to $v$ via a straight line
  is equal to $\swap(u,v)$.

  Let us now describe the construction of the corresponding
  \textsc{$2$-Embeddable $k$-Kemeny} instance $E$
  (see \Cref{fig:thm:hard:euclidean:2} for an illustration).
  For each $i \in \{0,1,\dots,n\} \cup \{2n\}$
  we put candidate $a_i$ in position $(4i,4i)$
  and candidate $b_i$ in position $(12n + 4i, 0)$.
  This gives us the total of $2n + 4$ candidates.

\begin{figure}
    \centering
    \begin{tikzpicture}
    \def\s{0.45} 
    \def\d{9 * \s}

    \draw[nearly transparent, gray] (0.5*\s - 1.5*\s, 0.5*\s + 1.5*\s) --
    (0.5*\s + 8.5*\s, 0.5*\s -8.5*\s);
    \draw[nearly transparent, gray] (1*\s - 2*\s, 1*\s + 2*\s) -- 
    (1.0*\s + 9*\s, 1.0*\s -9*\s);
    \draw[nearly transparent, gray] (1.5*\s - 2.5*\s, 1.5*\s + 2.5*\s) -- 
    (1.5*\s + 9.5*\s, 1.5*\s -9.5*\s);
    \draw[nearly transparent, gray] (2*\s - 3*\s, 2*\s + 3*\s) -- 
    (2.0*\s + 10*\s, 2.0*\s -10*\s);
    \draw[nearly transparent, gray] (2.5*\s - 3.5*\s, 2.5*\s + 3.5*\s) -- 
    (2.5*\s + 10.5*\s, 2.5*\s -10.5*\s);
    \draw[gray] (3*\s - 4*\s, 3*\s + 4*\s) -- 
    (3.0*\s + 11*\s, 3.0*\s -11*\s);
    \draw[gray] (3.5*\s - 3.5*\s, 3.5*\s + 3.5*\s) -- 
    (3.5*\s + 11.5*\s, 3.5*\s -11.5*\s);
    \draw[gray] (4*\s - 3*\s, 4*\s + 3*\s) -- 
    (4.0*\s + 12*\s, 4.0*\s -12*\s);
    \draw[gray] (4.5*\s - 2.5*\s, 4.5*\s + 2.5*\s) -- 
    (4.5*\s + 11.5*\s, 4.5*\s -11.5*\s);

    \draw[nearly transparent, gray] (\d + 0.5*\s, 7*\s) -- (\d + 0.5*\s,-8*\s);
    \draw[nearly transparent, gray] (\d + 1.0*\s, 7*\s) -- (\d + 1.0*\s,-8*\s);
    \draw[nearly transparent, gray] (\d + 1.5*\s, 7*\s) -- (\d + 1.5*\s,-8*\s);
    \draw[nearly transparent, gray] (\d + 2.0*\s, 7*\s) -- (\d + 2.0*\s,-8*\s);
    \draw[nearly transparent, gray] (\d + 2.5*\s, 7*\s) -- (\d + 2.5*\s,-8*\s);
    \draw[gray] (\d +3.0*\s, 7*\s) -- (\d +3.0*\s,-8*\s);
    \draw[gray] (\d +3.5*\s, 7*\s) -- (\d +3.5*\s,-8*\s);
    \draw[gray] (\d +4.0*\s, 7*\s) -- (\d +4.0*\s,-8*\s);
    \draw[gray] (\d +4.5*\s, 7*\s) -- (\d +4.5*\s,-8*\s);

    \draw[nearly transparent, gray] (\d/2 + 0.0*\s, 7*\s) -- (\d/2 + 0.0*\s, -8*\s);
    \draw[nearly transparent, gray] (\d/2 + 0.5*\s, 7*\s) -- (\d/2 + 0.5*\s, -8*\s);
    \draw[nearly transparent, gray] (\d/2 + 1.0*\s, 7*\s) -- (\d/2 + 1.0*\s, -8*\s);
    \draw[nearly transparent, gray] (\d/2 + 1.5*\s, 7*\s) -- (\d/2 + 1.5*\s, -8*\s);
    \draw[nearly transparent, gray] (\d/2 + 3.0*\s, 7*\s) -- (\d/2 + 3.0*\s, -8*\s);

    \def\x{\d/2 + 0.5*\s}
    \def\y{0.5*\s}
    \def\slope{8}
    \draw[nearly transparent, gray]
    (\x - \y/\slope - 8*\s/\slope, \y - \y - 8*\s) --
    (\x - \y/\slope + 7*\s/\slope, \y - \y + 7*\s);
    \def\x{\d/2 + 1*\s}
    \def\y{0.5*\s}
    \def\slope{9}
    \draw[nearly transparent, gray]
    (\x - \y/\slope - 8*\s/\slope, \y - \y - 8*\s) --
    (\x - \y/\slope + 7*\s/\slope, \y - \y + 7*\s);
    \def\x{\d/2 + 1.5*\s}
    \def\y{0.5*\s}
    \def\slope{10}
    \draw[nearly transparent, gray]
    (\x - \y/\slope - 8*\s/\slope, \y - \y - 8*\s) --
    (\x - \y/\slope + 7*\s/\slope, \y - \y + 7*\s);
    \def\x{\d/2 + 2*\s}
    \def\y{0.5*\s}
    \def\slope{11}
    \draw[nearly transparent, gray]
    (\x - \y/\slope - 8*\s/\slope, \y - \y - 8*\s) --
    (\x - \y/\slope + 7*\s/\slope, \y - \y + 7*\s);
    \def\x{\d/2 + 3.5*\s}
    \def\y{0.5*\s}
    \def\slope{14}
    \draw[nearly transparent, gray]
    (\x - \y/\slope - 8*\s/\slope, \y - \y - 8*\s) --
    (\x - \y/\slope + 7*\s/\slope, \y - \y + 7*\s);

    \def\x{\d/2 + 1*\s}
    \def\y{1*\s}
    \def\slope{3.5}
    \draw[nearly transparent, gray]
    (\x - \y/\slope - 8*\s/\slope, \y - \y - 8*\s) --
    (\x - \y/\slope + 7*\s/\slope, \y - \y + 7*\s);
    \def\x{\d/2 + 1.5*\s}
    \def\y{1*\s}
    \def\slope{4}
    \draw[nearly transparent, gray]
    (\x - \y/\slope - 8*\s/\slope, \y - \y - 8*\s) --
    (\x - \y/\slope + 7*\s/\slope, \y - \y + 7*\s);
    \def\x{\d/2 + 2*\s}
    \def\y{1*\s}
    \def\slope{4.5}
    \draw[nearly transparent, gray]
    (\x - \y/\slope - 8*\s/\slope, \y - \y - 8*\s) --
    (\x - \y/\slope + 7*\s/\slope, \y - \y + 7*\s);
    \def\x{\d/2 + 2.5*\s}
    \def\y{1*\s}
    \def\slope{5}
    \draw[nearly transparent, gray]
    (\x - \y/\slope - 8*\s/\slope, \y - \y - 8*\s) --
    (\x - \y/\slope + 7*\s/\slope, \y - \y + 7*\s);
    \def\x{\d/2 + 4*\s}
    \def\y{1*\s}
    \def\slope{6.5}
    \draw[nearly transparent, gray]
    (\x - \y/\slope - 8*\s/\slope, \y - \y - 8*\s) --
    (\x - \y/\slope + 7*\s/\slope, \y - \y + 7*\s);

    \def\x{\d/2 + 1.5*\s}
    \def\y{1.5*\s}
    \def\slope{2}
    \draw[nearly transparent, gray]
    (\x - \y/\slope - 8*\s/\slope, \y - \y - 8*\s) --
    (\x - \y/\slope + 7*\s/\slope, \y - \y + 7*\s);
    \def\x{\d/2 + 2*\s}
    \def\y{1.5*\s}
    \def\slope{2.333}
    \draw[nearly transparent, gray]
    (\x - \y/\slope - 8*\s/\slope, \y - \y - 8*\s) --
    (\x - \y/\slope + 7*\s/\slope, \y - \y + 7*\s);
    \def\x{\d/2 + 2.5*\s}
    \def\y{1.5*\s}
    \def\slope{2.666}
    \draw[nearly transparent, gray]
    (\x - \y/\slope - 8*\s/\slope, \y - \y - 8*\s) --
    (\x - \y/\slope + 7*\s/\slope, \y - \y + 7*\s);
    \def\x{\d/2 + 3*\s}
    \def\y{1.5*\s}
    \def\slope{3}
    \draw[nearly transparent, gray]
    (\x - \y/\slope - 8*\s/\slope, \y - \y - 8*\s) --
    (\x - \y/\slope + 7*\s/\slope, \y - \y + 7*\s);
    \def\x{\d/2 + 3.5*\s}
    \def\y{1.5*\s}
    \def\slope{4}
    \draw[nearly transparent, gray]
    (\x - \y/\slope - 8*\s/\slope, \y - \y - 8*\s) --
    (\x - \y/\slope + 7*\s/\slope, \y - \y + 7*\s);

    \def\x{\d/2 + 3*\s}
    \def\y{3*\s}
    \def\slope{0.5}
    \draw[nearly transparent, gray]
    (\x - \d/2 - 4*\s, \y - \d*\slope/2 - 4*\s*\slope) --
    (\x - \y/\slope + 7*\s/\slope, \y - \y + 7*\s);
    \def\x{\d/2 + 3.5*\s}
    \def\y{3*\s}
    \def\slope{0.666}
    \draw[nearly transparent, gray]
    (\x - \d/2 - 4.5*\s, \y - \d*\slope/2 - 4.5*\s*\slope) --
    (\x - \y/\slope + 7*\s/\slope, \y - \y + 7*\s);
    \def\x{\d/2 + 4*\s}
    \def\y{3*\s}
    \def\slope{0.833}
    \draw[nearly transparent, gray]
    (\x - \d/2 - 5*\s, \y - \d*\slope/2 - 5*\s*\slope) --
    (\x - \y/\slope + 7*\s/\slope, \y - \y + 7*\s);
    \def\x{\d/2 + 4.5*\s}
    \def\y{3*\s}
    \def\slope{1}
    \draw[nearly transparent, gray]
    (\x - \d/2 - 5.5*\s, \y - \d*\slope/2 - 5.5*\s*\slope) --
    (\x - \y/\slope + 7*\s/\slope, \y - \y + 7*\s);
    \def\x{\d/2 + 6*\s}
    \def\y{3*\s}
    \def\slope{1.5}
    \draw[nearly transparent, gray]
    (\x - \y/\slope - 8*\s/\slope, \y - \y - 8*\s) --
    (\x - \y/\slope + 7*\s/\slope, \y - \y + 7*\s);

    \draw[thick, blue] (0 * \s, 0 * \s) circle (1pt);
    \draw[thick, blue] (1 * \s, 1 * \s) circle (1pt);
    \draw[thick, blue] (2 * \s, 2 * \s) circle (1pt);
    \draw[thick, blue] (3 * \s, 3 * \s) circle (1pt);
    \draw[thick, blue] (6 * \s, 6 * \s) circle (1pt);
    \node at (0 * \s, 0.6 * \s) {\footnotesize $a_0$};
    \node at (1 * \s, 1.6 * \s) {\footnotesize $a_1$};
    \node at (2 * \s, 2.6 * \s) {\footnotesize $a_2$};
    \node at (3 * \s, 3.6 * \s) {\footnotesize $a_3$};
    \node at (6 * \s, 6.6 * \s) {\footnotesize $a_6$};

    \draw[thick, blue] (\d + 0*\s, 0*\s) circle (1pt);
    \draw[thick, blue] (\d + 1*\s, 0*\s) circle (1pt);
    \draw[thick, blue] (\d + 2*\s, 0*\s) circle (1pt);
    \draw[thick, blue] (\d + 3*\s, 0*\s) circle (1pt);
    \draw[thick, blue] (\d + 6*\s, 0*\s) circle (1pt);
    \node at (\d + 0*\s, 0.6*\s) {\footnotesize $b_0$};
    \node at (\d + 1*\s, 0.6*\s) {\footnotesize $b_1$};
    \node at (\d + 2*\s, 0.6*\s) {\footnotesize $b_2$};
    \node at (\d + 3*\s, 0.6*\s) {\footnotesize $b_3$};
    \node at (\d + 6*\s, 0.6*\s) {\footnotesize $b_6$};

    \fill[Green] (\d + 3.25*\s, -3.75*\s) circle (1pt);
    \fill[Green] (\d + 3.25*\s, -4.75*\s) circle (1pt);
    \fill[Green] (\d + 3.25*\s, -5.75*\s) circle (1pt);
    \fill[Green] (\d + 3.75*\s, -5.25*\s) circle (1pt);
    \fill[Green] (\d + 3.75*\s, -6.25*\s) circle (1pt);
    \fill[Green] (\d + 4.25*\s, -4.75*\s) circle (1pt);
    \fill[Green] (\d + 4.25*\s, -5.75*\s) circle (1pt);

    \draw[thick, gray] (\d + 3*\s, -3*\s) --
        (\d + 4.5*\s, -4.5*\s) --
        (\d + 4.5*\s, -7.5*\s) --
        (\d + 3*\s, -6*\s) --
        (\d + 3*\s, -3*\s);
    \end{tikzpicture}
    \caption{Illustration of the corresponding 2-Euclidean election to the graph instance in \Cref{fig:thm:hard:euclidean:1},
    as constructed in the proof of \Cref{thm:hard:euclidean}.
    Blue empty dots represent candidates,
    full green dots represent voters,
    and each line is a divider between a pair of candidates.
    Thicker fragments of the lines
    mark the edges of the grid structure.
    The dividers outside of the grid structure do not play a role in the construction
    (they are marked with lighter gray in the figure).}
    \label{fig:thm:hard:euclidean:2}
\end{figure}

  Next, let us analyze the dividers between all pairs of the candidates.
  For each $i \in \{0,\dots,n\}$,
  the divider $d_{b_{2n},b_i}$ is a vertical line passing through the point $(16n + 2i, 0)$.
  Analogously, for each $i \in \{0,\dots,n\}$,
  the divider $d_{a_{2n},a_i}$
  is a line with $45^\circ$ negative slope,
  passing through the point $(4n + 2i, 4n + 2i)$.
  Thus, we obtain a grid structure in which $n$ rows
  are divided by lines
  $d_{a_{2n},a_0},\dots,d_{a_{2n},a_n}$
  and $n$ columns by lines
  $d_{b_{2n},b_0},\dots,d_{b_{2n},b_n}$.
  The grid is bounded by $d_{a_{2n},a_0}$ from the bottom, $d_{a_{2n},a_n}$ from the top, $d_{b_{2n},b_0}$ from the left, and $d_{b_{2n},b_n}$ from the right.
  Let us show that no other divider passes through that area.

  The divider between each pair of candidates among $b_0, b_1,\dots,b_n$
  is a vertical line to the left of $d_{b_{2n},b_0}$.
  Similarly,
  the divider between each pair of candidates among $a_0, a_1\dots,a_n$
  is a parallel line to $d_{a_{2n},a_0}$ going below it.
  Thus, none of these dividers pass through the grid area.
  It remains to check the dividers between $a_i$ and $b_j$
  for each $i,j \in \{0,1,\dots,n,2n\}$.
  Fix arbitrary such $i$ and $j$,
  and observe that $d_{a_i,b_j}$ is
  either a vertical or a positively sloped line
  that passes through the middle point between $a_i$ and $b_j$
  that is $(6n + 2i + 2j, 2i)$.
  Since $6n + 2i + 2j \le 6n + 4n + 4n = 14n < 16n$,
  divider $d_{a_i,b_j}$ crosses $d_{b_{2n},b_0}$,
  (i.e., the left bound of the grid), above point $(16n,2i)$
  (or not crosses it at all and is always to the left of $d_{b_{2n},b_0}$).
  On the other hand, $d_{a_{2n},a_n}$, so the top bound of the grid,
  crosses $d_{b_{2n},b_0}$ in the point $(16n, -4n)$.
  Thus, divider $d_{a_i,b_j}$ does not pass through the grid area
  for each $i,j \in \{0,1,\dots,n,2n\}$.

  Finally, we place one voter for each vertex of the graph
  in the $\textsc{GridGraphDominatingSet}$ instance
  in a cell of our grid
  corresponding to its cell in the $n \times n$ Cartesian grid.
  Specifically, for a given vertex $(i,j) \in V \subseteq [n]^2$,
  we place a single voter $v_{i,j}$
  in position $(16n + 2i -1, -8n + 4j - 2i - 1)$.
  This is to the right of $d_{b_{2n},b_{i-1}}$,
  and to the left of $d_{b_{2n},b_{i}}$
  that cover all points with the first coordinate being
  $16n + 2i - 2$ and $16n + 2i$, respectively.
  Also, this is above $d_{a_{2n},a_{j-1}}$
  and below $d_{a_{2n},a_j}$,
  which go through the points
  \begin{align*}
      &(16n + 2i -1, -8n + 4(j-1) - 2i + 1) \quad \mbox{and} \\
      &(16n + 2i -1, -8n + 4j - 2i + 1),
  \end{align*}
  respectively.
  Finally, we set $k=s$ and $q = |V| - k$.

  In the remainder of the proof,
  we show that there exists a dominating set of size $s$ in graph $G$, if and only if,
  there is a set $R$ of $k$ rankings embeddable in our instance
  such that $k\hbox{-}\kemeny_E(R) \le q$.

  If there exists a dominating set $S$,
  then we can select rankings given by voters that correspond to vertices in $S$, i.e., $R = \{v_{i,j} : (i,j) \in S\}$.
  Clearly, for every voter in $R$,
  the distance to the closest ranking in $R$ is $0$.
  Moreover, for every other voter, $v_{i,j} \not \in R$,
  there is voter $u$ that corresponds to vertex in $S$
  that is neighboring node $(i,j)$.
  But this means
  that there is only a single divider between $u$ and $v_{i,j}$,
  hence  $\swap(u,v_{i,j}) = 1$.
  Therefore, we obtain that
  $k\hbox{-}\kemeny_E(R) = |V| - |R| = |V| - k = q$.

  For the other direction,
  assume that there is a set of $k$ rankings $R$
  such that $k\hbox{-}\kemeny_E(R) \le q$.
  Observe that it has to mean
  that the rankings of at least $k$ voters are selected to $R$.
  Otherwise, we would have at least $|V| - k + 1$ voters
  with distance at least 1, which would yield
  $k\hbox{-}\kemeny_E(R) \ge |V| - k + 1 = q + 1$.
  On the other hand, every other voter, $v \not \in R$,
  has to be at distance 1 to some ranking in $R$.
  Otherwise, if we have at least a single voter
  with a distance of at least $2$ to every ranking in $R$,
  then again $k\hbox{-}\kemeny_E(R) \ge (|V| - k - 1) + 2 = q + 1$.
  Thus, for every $v \in V \setminus R$
  there must exist voter $u \in R$
  such that $v$ and $u$ belong to neighboring cells.
  But in such a case,
  $S = \{(i,j) : v_{i,j} \in R\}$ is a dominating set in the original instance,
  which concludes the proof.
}

\subsection{Algorithms for Condorcet Domains}

\citet{fal-kac-sor-szu-was:c:div-agr-pol-map} have shown that for
every $\varepsilon > 0$ there is an $\fpt$ approximation algorithm for
\kkemeny with $1+\varepsilon$ approximation ratio, parameterized by
the number $n$ of the voters and running in time $O^*(n^n)$.
Using dynamic programming, for Condorcet domains we improve
this to an exact $\fpt$ algorithm, running in time
$O^*(3^n)$. For the general domain such a result
seems impossible, as
$1$-Kemeny is $\np$-hard even for $n=4$
voters~\citep{dwo-kum-nao-siv:c:rank-aggregation,bie-bra-den:j:kemeny-hardness}.

\begin{theorem}\label{thm:fpt-n}
  There exists an $\fpt$ algorithm parameterized by the number $n$ of the
  votes that given an instance of $\kkemeny$---where the votes come
  from a Condorcet domain---solves it in time $O^*(3^n)$.
\end{theorem}
\appendixproof{Theorem}{thm:fpt-n}{
  Consider an instance of $\kkemeny$ with election $E = (C,V)$ and
  integers $k$ and $q$. Let $n = |V|$; w.l.o.g., we assume that
  $n > k$ (otherwise there is a trivial $k$-Kemeny set with score
  equal to $0$).  The election comes from a Condorcet domain so, as a
  consequence, for every subsequence $V'$ of its voters election
  $(C,V')$ has at least one Condorcet ranking $r(V')$ which can be
  computed in polynomial time and which also is a Kemeny ranking for
  $(C,V')$ (recall the discussion in \Cref{sec:prelim}). Hence the job
  of our algorithm is to check if it is possible to partition $V$ into
  $V_1, \ldots, V_k$ such that the sum of Kemeny scores of election
  $(C,V_1), \ldots, (C,V_k)$ is at most $q$.

  To this end, we use dynamic programming based on computing the
  following function. For a subsequence $V'$ of voters and a positive
  integer $t$, we define $f(V',t)$ to be $t$-Kemeny score of election
  $(C,V')$. Our goal is to check if $f(V,k) \leq q$. As a base case,
  we can compute $f(V',1)$ for every $V' \subseteq V$, by evaluating
  the Kemeny score of the Condorcet ranking for $(C,V')$. We compute
  $f$ for $t = 2, 3, \ldots k$ in $k-1$ iterations as follows (let us
  fix the iteration number $t$ that we execute):
  \begin{enumerate}
  \item Initially, we set $f(V',t) = \infty$ for each $V' \subseteq V$. We will be
    updating these values throughout the iteration.
  \item We consider all functions $g \colon V \rightarrow \{0,1,2\}$.
  \item For each function $g$, we partition $V$ into collections
    $V^{(0)}$, $V^{(1)}$, and $V^{(2)}$, so that each voter $v \in V$
    belongs to this $V^{(g(v))}$. If the current value we have stored
    for $f(V^{(1)} \cup V^{(2)},t)$ is greater than
    $f(V^{(1)},1) + f(V^{(2)},t-1)$ then we update
    $f(V^{(1)} \cup V^{(2)},t)$ to be this value (intuitively, this
    corresponds to a scenario where there is a $t$-Kemeny set such
    that voters from $V^{(1)}$ are closest to one ranking, namely the
    Condorcet ranking of $(C,V^{(1)})$, and each voter in $V^{(2)}$ is
    closest to one of $t-1$ other rankings from this $t$-Kemeny set).
  \end{enumerate}
  We accept if thus-computed value $f(V,k)$ is at most $q$ and reject
  otherwise.  The correctness of the algorithm is immediate. The
  running time follows from the fact that in each iteration we need to
  consider $3^n$ functions $g$ and for each of them the algorithm runs
  in time $O(|C|+|V|)$. Thus, altogether, the running time is
  $O(3^n \cdot t \cdot (|C|+|V|))$; note that this time is also
  sufficient to compute Condorcet rankings for all subsets of voters.
}

For single-crossing elections we even get  a polynomial-time
algorithm.  The idea is to model $\kkemeny$ as a Chamberlin--Courant
(CC) multiwinner election~\citep{cha-cou:j:cc}, where each vote from
the original election is both a candidate and a voter in the CC one,
ranking itself and the others with respect to the swap distance from
itself. This yields single-peaked election, for which CC is
tractable~\citep{bet-sli-uhl:j:mon-cc,sor-vas-xu:c:cc-sp}.

\begin{theorem}\label{thm:sc-algorithm}
    Single-crossing instances of \kkemeny can be solved in time $\tilde{O}(nm+n^2)$.
\end{theorem}
\appendixproof{Theorem}{thm:sc-algorithm}{ 
We provide an algorithm which finds a $k$-Kemeny set, i.e., the set that minimizes the $k$-Kemeny score.
This is enough to solve the decision variant of \kkemeny.
Let $E = (C,V)$ be a single-crossing instance of \kkemeny and let $(v_1, \dots, v_n)$ be an ordering of voters realizing the single-crossing property.
The ordering can be found in polynomial time~\citep[Theorem 4.8]{elk-lac-pet:t:restricted-domains-survey}.

First, we observe that we can look for the $k$-Kemeny set among the votes.

\begin{lemma}\label{lemma:kkemeny-sc-among-votes}
  There exists an optimal solution to \kkemeny on a single-crossing election that is a subset of $V$.
\end{lemma}
\begin{proof}
    Let $R$ be any fixed optimal solution.
    By the definition of $k\hbox{-}\kemeny_E(R)$, every voter from $V$ is assigned to a ranking in $R$ (specifically, to the one that minimizes the swap distance to $R$).
    Therefore, we can define a $k$-partition of $V$ into $V_1, \dots, V_k$,
    where voters from $V_i$ are assigned to the same $r_i \in R$.
    
    We observe that, for every $i \in [k]$, $r_i$ is an optimal solution to an instance $(C,V_i)$ of \onekemeny.
    Otherwise we could decrease the $k$-Kemeny score of $R$ by replacing $r_i$ with an optimal solution to \onekemeny on $(C,V_i)$---this would be in contradiction with optimality of $R$.
    
    Moreover, we observe that $(C,V_i)$ is a single-crossing instance (with respective suborder of voters) and since $(C,V_i)$ has rankings from a Condorcet domain it contains at least one Condorcet ranking which we call $o_i \in V_i$.
    \Cref{cor:kemeny-condorcet} implies that $o_i$ is an optimal solution to \onekemeny on $(C,V_i)$.
    It means that $O = \{o_1, \dots, o_k\}$ is a solution to $E$ and has the $k$-Kemeny score equal to
    \begin{align*}
      \textstyle\sum_{i \in [k]}\sum_{v \in V_i} \min_{o \in O}\swap(v,o)
      \leq \textstyle\sum_{i \in [k]} \sum_{v \in V_i} \swap(v,o_i)
      = \sum_{i \in [k]} \sum_{v \in V_i} \swap(v,r_i),
    \end{align*}
    where the inequality follows from the fact that on the LHS a voter $v \in V_i$ is assigned to nearest $o \in O$, but on the RHS it is assigned to concrete $o_i$.
    The last one term equals $k\hbox{-}\kemeny_E(R)$
    (by optimality of $R$ the inequality becomes an equality).
    Therefore, $O$ is a requested solution by the lemma statement.
\end{proof}

Using \Cref{lemma:kkemeny-sc-among-votes} we will solve the problem by considering a problem of choosing size-$k$ subset $W$ of $V$,
where the objective function is minimizing the sum of $\swap$ distances between elements of $V$ and $W$.
This can be done by solving a well known Chamberlin-Courant (CC) multiwinner election problem~\citep{cha-cou:j:cc,bet-sli-uhl:j:mon-cc,sor-vas-xu:c:cc-sp}.
Formally, we create an instance $E' = (V',C',\rho,k)$ of CC,
in which we want to find a size-$k$ subset $W \subseteq C'$ such that the \emph{total misrepresentation} of $W$ defined as
$\rho(W) = \sum_{v \in V'} \rho(v, W) = \sum_{v \in V'} \min_{c \in W} \rho(v, c)$
is minimized, and in our case we have:
\begin{itemize}
    \item $V' = V$, i.e., $E'$ has the same set of voters as $E$.
    \item $C' = V$, i.e., we select $k$ winners from the set of candidates $C'$ corresponding the set of voters in $E$.
    \item A misrepresentation function $\rho: V' \times C' \to \reals_{\geq 0}$ is defined for every $v_i \in V', c_j \in C'$ as:
    $$
      \rho(v_i, c_j) = \swap(v_i, c_j).
    $$
\end{itemize}
Therefore, we can define the new instance of CC as $E' = (V,V,\swap,k)$ but we will be using the terms $V',C',\rho$ for clearer relation to CC.
Formally, we have a correspondence between solutions in both instances (hence, also between optimal solutions of size $k$).
\begin{lemma}\label{lemma:kkemeny-sc-equi-cc}
    For every subset of rankings $R \subseteq V$ 
    we have
    $$k\hbox{-}\kemeny_E(R) = \rho(R).$$
\end{lemma}
\begin{proof}
    Using definitions of the swap distance and $\rho$ we have
    \begin{align*}
        k\hbox{-}\kemeny_E(R)
        &= \textstyle\sum_{v \in V} \min_{r \in R}\swap(v,r)\\
        &= \textstyle\sum_{v \in V} \min_{r \in R}\rho(v,r)  = \rho(R).\qedhere
    \end{align*}
\end{proof}

Moreover, we observe that $E'$ is actually a single-peaked instance of CC.
For that we will use the following definition of single-peakedness~\citep[Proposition~3]{bet-sli-uhl:j:mon-cc}.
\begin{definition}\label{def:sp}
    $(V^*,C^*,\rho^*)$ is called a single-peaked preference profile with an axis $c_1^* \lhd c_2^* \lhd \cdots \lhd c_m^*$ (a linear order $\lhd$ over $C^*$)
    if for every triple of distinct candidates $c_i^*,c_j^*,c_k^* \in C^*$ with
    $c_i^* \lhd c_j^* \lhd c_k^*$ or with $c_k^* \lhd c_j^* \lhd c_i^*$ we have the following implication for every voter $v^* \in V^*$:
    $$
      \rho^*(v^*,c_i^*) < \rho^*(v^*,c_j^*)
      \implies \rho^*(v^*,c_j^*) \leq \rho^*(v^*,c_k^*).
    $$
\end{definition}

\begin{lemma}\label{lemma:cc-is-sp}
    $(V',C',\rho)$ is a single-peaked preference profile with an axis $v_1 \lhd v_2 \lhd \cdots \lhd v_n$.
\end{lemma}
\begin{proof}
    We need to check the condition (the implication) from \Cref{def:sp} for every voter $v_x \in V = V'$
    and for every triple of distinct candidates $v_i, v_j, v_k \in V = C'$ with $v_i \lhd v_j \lhd v_k$ or with $v_k \lhd v_j \lhd v_i$.
    Below we consider only the case $v_i \lhd v_j \lhd v_k$ as the proof for the other case is analogous.
    
    If $\rho(v_x,v_i) < \rho(v_x,v_j)$ then by the definition of $\rho$ we have
    $\swap(v_x,v_i) < \swap(v_x,v_j)$.
    We consider two cases depending on how $v_x$ is positioned in an order $v_1 \lhd v_2 \lhd \cdots \lhd v_n$ comparing to $v_j$.
    
    The first case is $v_x \lhd v_j \lhd v_k$.
    By the definition of single-crossingness on $(v_1,\dots v_n)$ we have
        \begin{align*}
          \swap(v_x,v_k)
          &= |\{(a,b) \in C \times C : x \leq t_{ab} < k \}|\\
          &\geq |\{(a,b) \in C \times C : x \leq t_{ab} < j \}|\\
          &= \swap(v_x,v_j).  
        \end{align*}
    By the equivalence of $\swap$ and $\rho$ we obtain $\rho(v_x,v_j) \leq \rho(v_x,v_k)$ as requested.
        
    The above might be derived from observations made by~\citet[Proposition~4.7]{elk-lac-pet:c:restricted-domains} regarding the relation of single-craziness of $v_1 \lhd v_2 \lhd \cdots \lhd v_n$ and $\swap$ distance between $v_1$ and $v_i$ and between $v_i$ and $v_{i+1}$.

    The second case is $v_i \lhd v_j \lhd v_x$.
    By the definition of single-crossingness on $(v_1,\dots v_n)$ we have
        \begin{align*}
          \swap(v_i,v_x)
          &= |\{(a,b) \in C \times C : i \leq t_{ab} < x \}|\\
          &\geq |\{(a,b) \in C \times C : j \leq t_{ab} < x \}|\\
          &= \swap(v_j,v_x).  
        \end{align*}
    By the equivalence of $\swap$ and $\rho$ we obtain $\rho(v_x,v_i) \geq \rho(v_x,v_j)$,
    but this is in contradiction with the initial assumption $\rho(v_x,v_i) < \rho(v_x,v_j)$.

    This finishes the proof of the lemma.
\end{proof}

Finally, \Cref{lemma:kkemeny-sc-equi-cc} together with \Cref{lemma:cc-is-sp} imply that in order to provide a solution to \kkemeny on single-crossing elections it is enough to select an optimal committee under the Chamberlin-Courant rule on single-peaked elections.
This can be done in polynomial-time using known algorithms, e.g.,
a classic dynamic program of~\citet[Theorem 8]{bet-sli-uhl:j:mon-cc} with the running time $O(|V'| \cdot |C'|^2)=O(n^3)$ or
an almost linear (in the input size, i.e. , $\tilde{O}(|V'| \cdot |C'|) = \tilde{O}(n^2)$) time algorithm of~\citet[Theorem 6]{sor-vas-xu:c:cc-sp}.\footnote{
The algorithm has running time $\tilde{O}(|V'| \cdot |C'|) = \tilde{O}(n^2)$ ($\tilde{O}$ hides subpolynomial factors) assuming that the single-peaked order of voters is given in the input.
In the instance we have constructed each preference from $(V',C',\rho)$---after removing duplicates from $V'$ (i.e. having the same vote in $E$)---is a linear order over $C'$, hence we can recognize an axis in $O(|V'| \cdot |C'|) = O(n^2)$ time~\citep{esc-lan-ozt:c:single-peaked-consistency} so this does not increase the claimed running time.
}

Unfortunately, both algorithms require an oracle access to values of $\rho$ therefore we need to define its $O(n^2)$ many values.
It can be done using a straightforward computation of the $\swap$ distance in time $O(m\sqrt{\log m})$~\citep{cha-pat:c:fast-swap-distance} for every pair of voters from $V$.
In total, this would take $O(n^2 \cdot m\sqrt{\log m})$ time for computing $\rho$ and such an approach would be a bottle-neck of the whole algorithm.

A more refined approach is as follows.
First, we find an SC ordering of voters $(v_1,\dots,v_n)$ in time $O(nm\log m)$~\citep[Theorem 4.8]{elk-lac-pet:t:restricted-domains-survey}.
Then, we compute $n-1$ swap distances between each consecutive voters $v_i$, $v_{i+1}$ in time $O(nm\sqrt{\log m})$~\citep{cha-pat:c:fast-swap-distance}.
These values allow us to compute all values of $\swap(v_1, v_i)$ in additional time $O(n)$ using a formula $\swap(v_1, v_{i+1}) = \swap(v_1, v_i) + \swap(v_i, v_{i+1})$~\citep[Proposition~4.7]{elk-lac-pet:c:restricted-domains}.
Next, we can compute all $O(n^2)$ values of $\rho = \swap$ using a formula $\swap(v_{i+1},v_j) = \swap(v_i,v_j) - \swap(v_i, v_{i+1})$, each in constant time.
Therefore, computing all values of $\rho$ can be done in $O(nm\log m + n^2)$ time.

The total running time of our algorithm via finding an SC ordering, computing all values of $\rho$ and applying an algorithm of \citet{sor-vas-xu:c:cc-sp} is $\tilde{O}(nm + n^2)$.
}

\section{Diversity of Structured Domains}\label{sec:diversity-experiments}
Next we move on to the diversity analysis of our domains and the
elections that one can sample from them. 
Given an election $E$ and an integer $k$, we write $\kappa_E(k)$ to
denote its $k$-Kemeny score, i.e., the $k$-Kemeny score of its
$k$-Kemeny set.  In this section we often view domains as elections that include a
single copy of each possible vote.

\citet{fal-kac-sor-szu-was:c:div-agr-pol-map,fal-mer-nun-szu-was:c:map-dap-top-truncated}
proposed to measure the diversity of an election
$E = (C,V)$ using function $D$ defined as follows
($w_1 \geq w_2 \geq \cdots$ are weights and $N(E)$ is a
normalizing factor, depending on $|C|$ and $|V|$):
\[
  D(E) = N(E)\cdot \big( w_1 \kappa_E(1) + w_2 \kappa_E(2) + w_3 \kappa_E(3) + \cdots \big);
\]
the larger is the value $D(E)$, the more diverse is $E$.
However, the choice of $w_1, w_2, \ldots$ and $N(E)$ is not obvious
and the two above-cited papers make different ones. We follow their
general approach, but focusing on qualitative comparisons between
diversities of various elections/domains, by analyzing vectors
\[
  \kappa(E) = \left(\frac{1}{|V|}\kappa_E(1), \ldots,
  \frac{1}{|V|}\kappa_E(|C|)\right).
\]
We refer to values $\frac{1}{|V|}\kappa_E(k)$ as normalized $k$-Kemeny
scores.  We view an election $E' = (C,V')$ as more diverse than
election $E'' = (C,V'')$ if $\kappa(E')$ dominates
$\kappa(E'')$---i.e., has greater-or-equal values on each
coordinate---occasionally arguing how to resolve the situation when
dominance does not hold either way.
\citet{fal-kac-sor-szu-was:c:div-agr-pol-map} also suggested that
$\kappa_E(1)-\kappa_E(2)$, normalized appropriately, is a good measure
of $E$'s polarization.

\subsection{Setup, Domains, and Statistical Cultures}\label{sec:setup-experiments}
Throughout this section, we almost exclusively focus on the case of
$m=8$ candidates. This suffices for many realistic settings
and is small enough for efficient computations and clear visualizations.
Unless specified otherwise, all elections sampled from statistical
cultures consist of $512$ voters.  All experiments calculating the
$k$-Kemeny scores or domain sizes are averaged across $100$ samples.

Let us fix $C = \{c_1, \ldots, c_8\}$. We focus on the following nine
domains, explained below (some of the acronyms were already introduced
in \Cref{sec:prelim}):
\begin{quote}
  \text{1D-Int.}, \text{2D-Sq.}, \text{3D-Cb.}, \text{SC}, 
                     \text{SP}, \text{SP/DF}, \text{SPOC}, \text{GS/bal}, \text{GS/cat}.
\end{quote}
\paragraph{Euclidean Domains.}
For positive integer $t$, $t$D-Hyper\-Cube domain is a Euclidean
domain where candidate points are selected uniformly at random from
$[-1,1]^t$ (so, in fact, this is a family of domains, one for each embedding function).
Note that these domains include all preference orders that can be
obtained by putting a voter point somewhere in $\reals^t$, even if
this location is outside of $[-1,1]^t$. For $t \in \{1,2,3\}$, we refer
to these domains as 1D-Interval, 2D-Square, and 3D-Cube, respectively
(abbreviated as 1D-Int., 2D-Sq., and 3D-Cb.).

\paragraph{Single Crossing Domains.} To obtain a single-crossing
domain (SC), we follow the approach of
\citet{szu-boe-bre-fal-nie-sko-sli-tal:j:map}: We form a sequence
$v_1, \ldots, v_{\binom{m}{2}}$ of votes, where
$v_1 \colon c_1 \pref \cdots \pref c_8$ and for each $i > 1$ we obtain
$v_i$ from $v_{i-1}$ by swapping a pair of candidates $c_p, c_q$,
$p < q$, that are ranked consecutively (we select such a pair
uniformly at random).

\paragraph{SP/DF and the Remaining Domains.} The remaining five domains are unique
up to renaming the candidates and, with the exception of SP/DF, were
described in \Cref{sec:prelim}.  The SP/DF domain (or, the
SP/double-forked domain) is a domain of votes that are single-peaked
on the following tree:
\begin{center}
  \begin{tikzpicture}
          \node at (-0.6, 0.35) {$c_1$};
          \node at (-0.6, -0.35) {$c_2$};
          \node at (0, 0) {$c_3$};
          \node at (0.8, 0) {$c_4$};
          \node at (1.6, 0) {$c_5$};
          \node at (2.4, 0) {$c_6$};
          \node at (3, 0.35) {$c_7$};
          \node at (3, -0.35) {$c_8$};
          \draw (-0.4,0.3) -- (-0.2,0.1);
          \draw (-0.4,-0.3) -- (-0.2,-0.1);
          \draw (0.2,0) -- (0.6,0);
          \draw (1.0,0) -- (1.4,0);
          \draw (1.8,0) -- (2.2,0);
          \draw (2.6,0.1) -- (2.8,0.3);
          \draw (2.6,-0.1) -- (2.8,-0.3);
  \end{tikzpicture}
\end{center}

\paragraph{Sizes of the Domains.}
In \Cref{fig:domain_sizes} we plot the sizes of our domains---i.e.,
the numbers of distinct votes included in each---depending on the
number of candidates (note that the $y$ axis has a logarithmic scale).
More details can be found in \Cref{app:domain-sizes}.
Interestingly, even though the sizes of SP (which is equal to the size of GS), SP/DF, and SPOC grow
exponentially and the sizes of Euclidean domains grow polynomially,
for $m=8$, 2D-Square is larger than SP (GS) and almost as large as SPOC
and SP/DF. 3D-Cube is outright larger and remains so until $m=16$.

\begin{figure}[t]
    \centering
    \begin{minipage}[t]{0.45\columnwidth}
        \centering
        \includegraphics[width=0.99\linewidth]{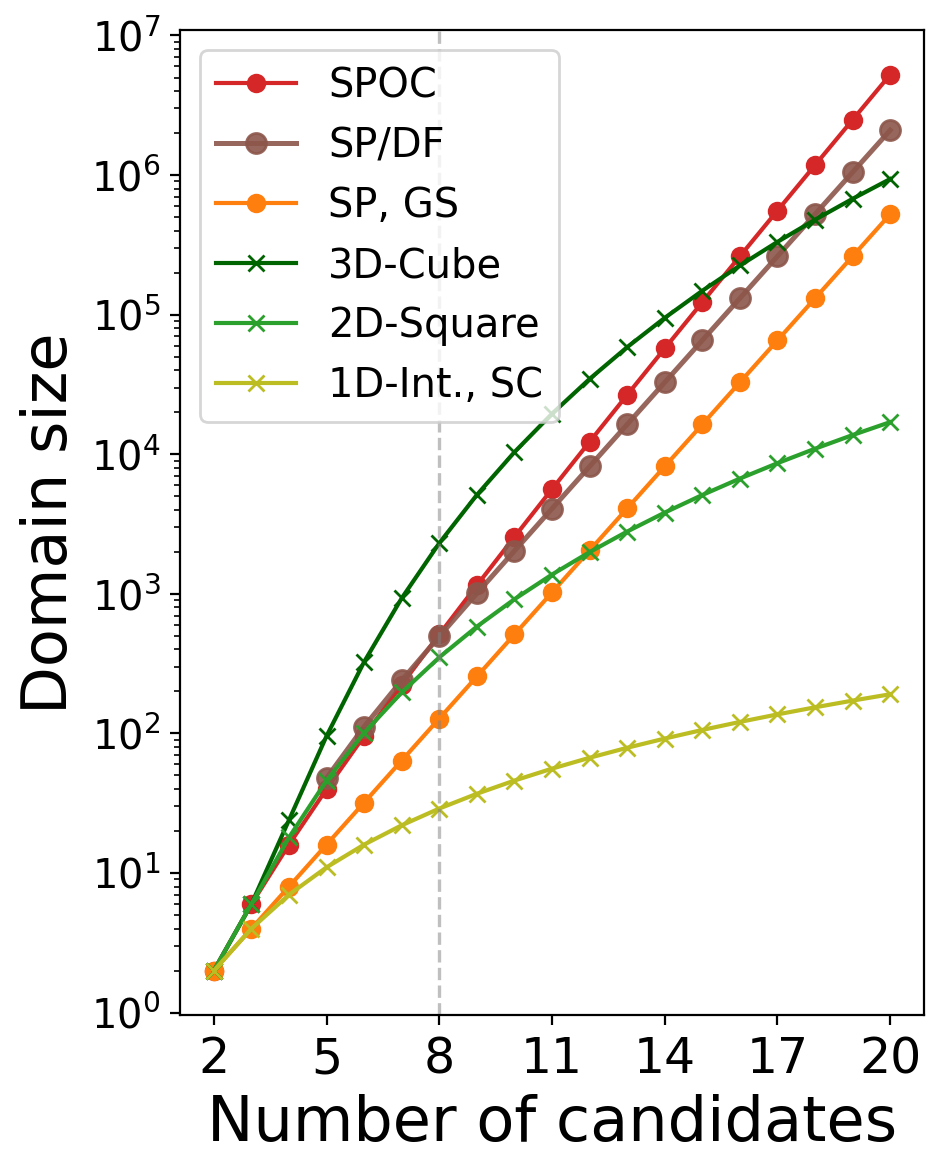}
        \caption{\label{fig:domain_sizes} Domain size as a function of the number of candidates.}
    \end{minipage}
    \hspace{0.2cm}
    \begin{minipage}[t]{0.442\columnwidth}
        \centering
        \includegraphics[width=0.99\linewidth]{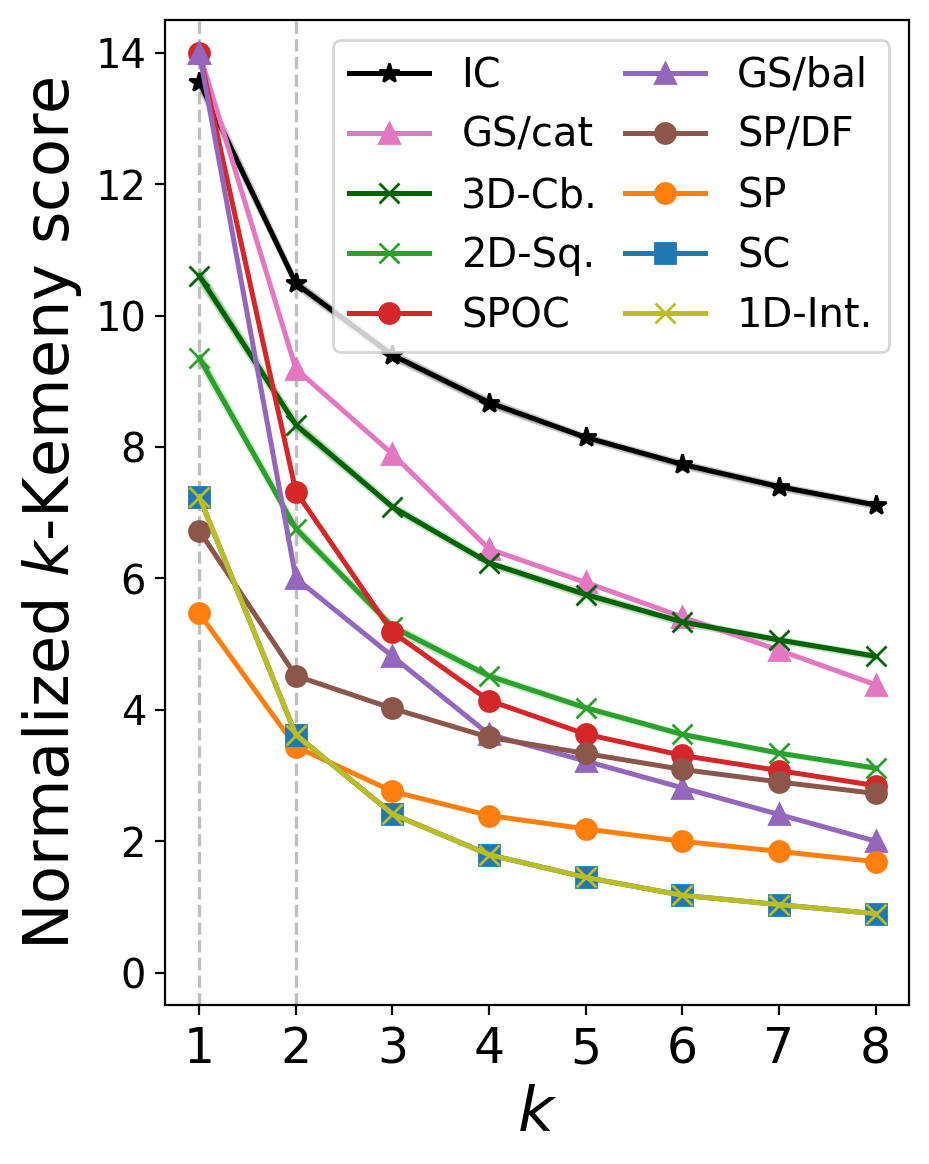}
        \caption{\label{fig:diversity} Average $\kappa(E)$ as a function of $k$.}
    \end{minipage}%
\end{figure}

\paragraph{Statistical Cultures.} Whenever we generate an
election with votes from a given domain, we do so by sampling the
required number of votes using a respective statistical culture. In
particular, for each of our domains we consider its variant of
impartial culture, i.e., sampling votes uniformly at random. For SP,
this is known as the Walsh model~\citep{wal:t:generate-sp}.
We also consider the Conitzer model of sampling SP
votes~\citep{con:j:eliciting-singlepeaked}:
First, we choose candidate $c_i$ uniformly at random and place it
on top of the vote.  Next, we perform $|C|-1$ iterations, so that in
each we extend by one candidate, selected uniformly at random among at
most two, so that the top-ranked candidates in the generated vote form
an interval with respect to the axis (we assume
$c_1 \lhd \cdots \lhd c_8$ for both models).
For $t$D-HyperCube domains, we use the $r$-Box model, $r \in \reals$,
which works as follows:
Given a domain and its embedding function, to sample a vote we choose
its point uniformly from $[-r, r]^t$
(recall that candidates are chosen from $[-1,1]^t$).
Large fraction of literature
in computational social choice that uses the $t$D-HyperCube models
uses the $1$-Box model (or an isomorphic one) and not impartial
culture over the given
domain~\citep{boe-fal-jan-kac-lis-pie-rey-sto-szu-was:c:guide}.

\paragraph[Computing k-Kemeny]{Computing $\boldsymbol{k}$-Kemeny.}
Given an election $E$, an integer $k$, and a search space (i.e., a
list of votes) we compute $\kappa_E(k)$ using a local search approximation algorithm
of \citet{fal-kac-sor-szu-was:c:div-agr-pol-map}. We start with a set
of $k$ centers, i.e., 
rankings selected uniformly at random from the search space,
and iteratively improve this set by replacing one of the
centers with a better one, from the search space, until no improvement
is found.  For Condorcet domains, it suffices to search over the
domain itself~\citep{bar:j:kemeny-condorcet,tru:t:kemeny-condorcet}.
For non-Condorcet domains, the search space consists of all the votes
from the domain and of $512$ votes sampled from IC (using more IC votes
did not lead to notable improvements).
We perform
10 random starts and select the best outcome.

\subsection{Microscope View of Our Domains}

\begin{figure}[t]
    \centering
    \includegraphics[width=0.9\linewidth, trim=0 50 0 0, clip]{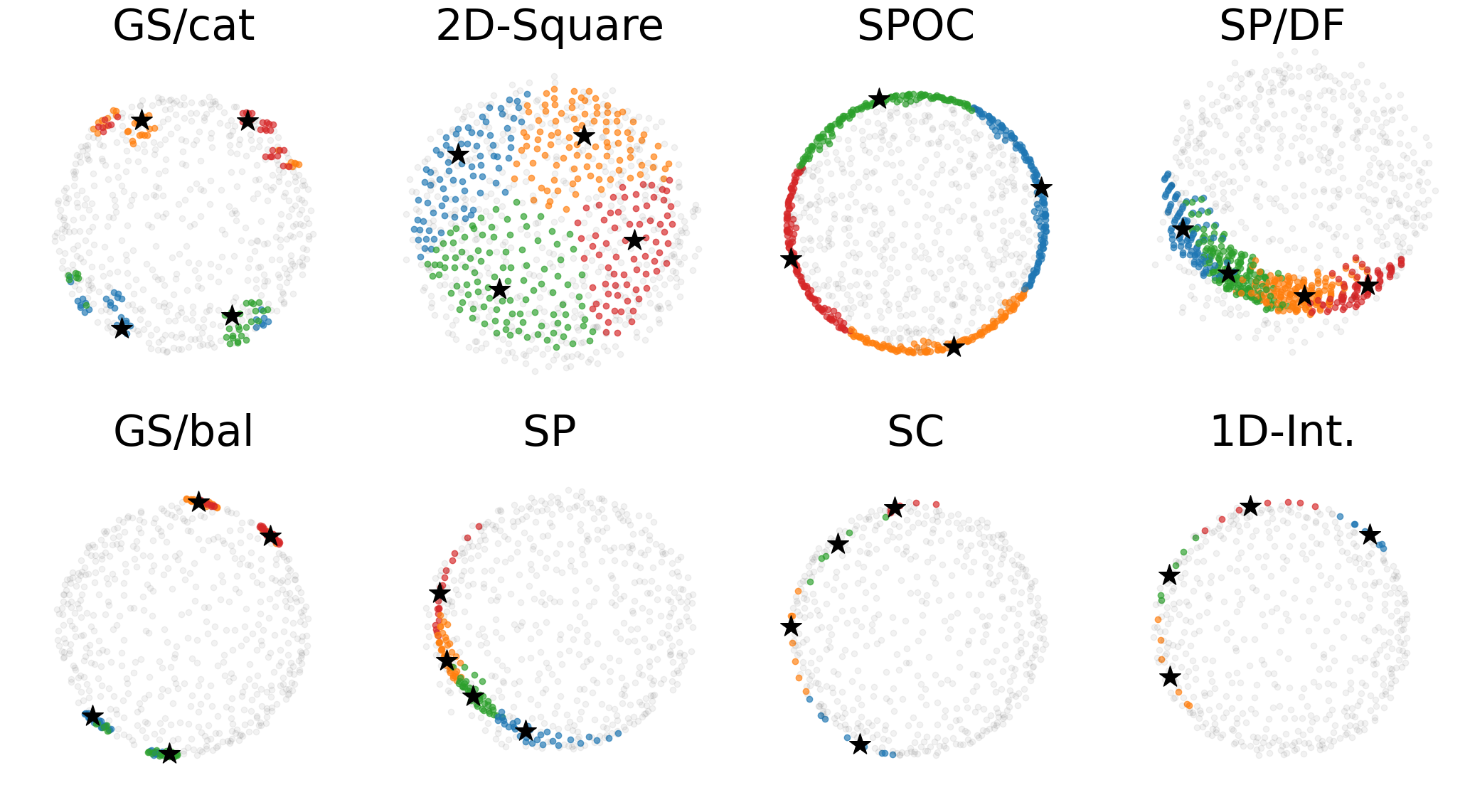}
    \caption{Microscopes of our domains with IC.}
    \label{fig:microscope_with_ic}
\end{figure}

Let us fix some election $E = (C,V)$, with $V = (v_1, \ldots, v_n)$.
Following
\citet{fal-kac-sor-szu-was:c:div-agr-pol-map},
we visualize $E$ as follows: We depict each vote $v \in V$ as a point
on a plane, so that the Euclidean distance between points
corresponding to votes $v_i, v_j \in V$ is as similar to
$\swap(v_i,v_j)$ as possible; to ensure such an embedding, we use the
multidimensional scaling algorithm~\citep{kru:j:mds,lee:j:mds}. While
not perfectly accurate, such visualizations---which we call
\emph{microscopes}---offer some intuition of the nature of the votes
in the election. They follow the general ideas of the map framework of
\citet{szu-boe-bre-fal-nie-sko-sli-tal:j:map}.

In~\Cref{fig:microscope_with_ic}, we present one microscope for each
of our domains, except 3D-Cb. (for 1D-Int., 2D-Sq., and SC we show one
representative example of a domain). The microscope shows the domain
with
$512$ additional votes generated from IC (light gray dots);
microscopes without added IC votes are in the appendix.
For each domain, we have computed an approximate $4$-Kemeny set, we
indicate the rankings from these sets with stars, and we colored the
votes according to the closest one.
As expected, votes of the same
color cluster together, providing a sanity check.
We refer to these plots in the following discussions.

\begin{figure}[t]
    \centering
    \includegraphics[width=0.2826\linewidth]{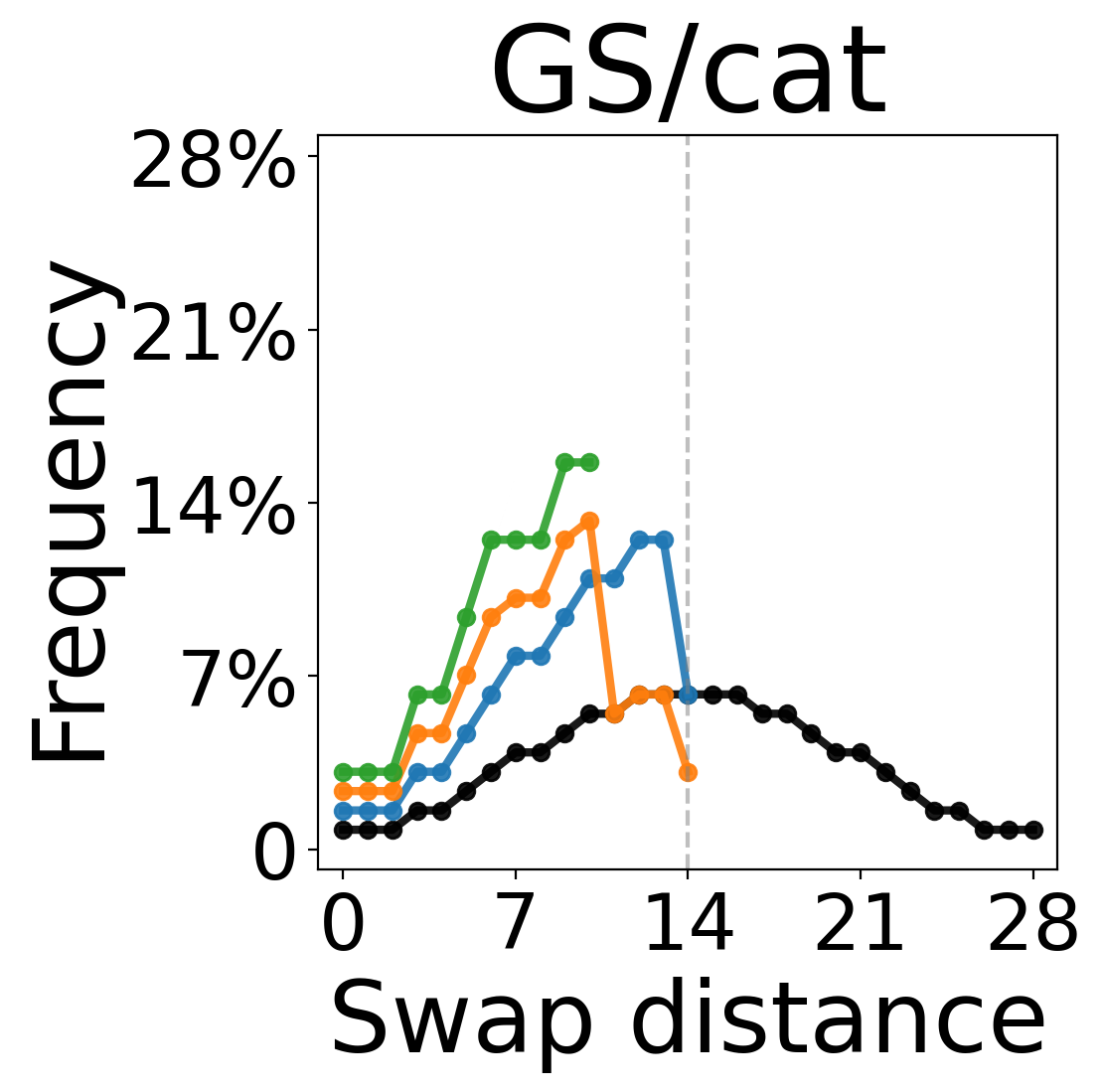}%
    \includegraphics[width=0.20592\linewidth]{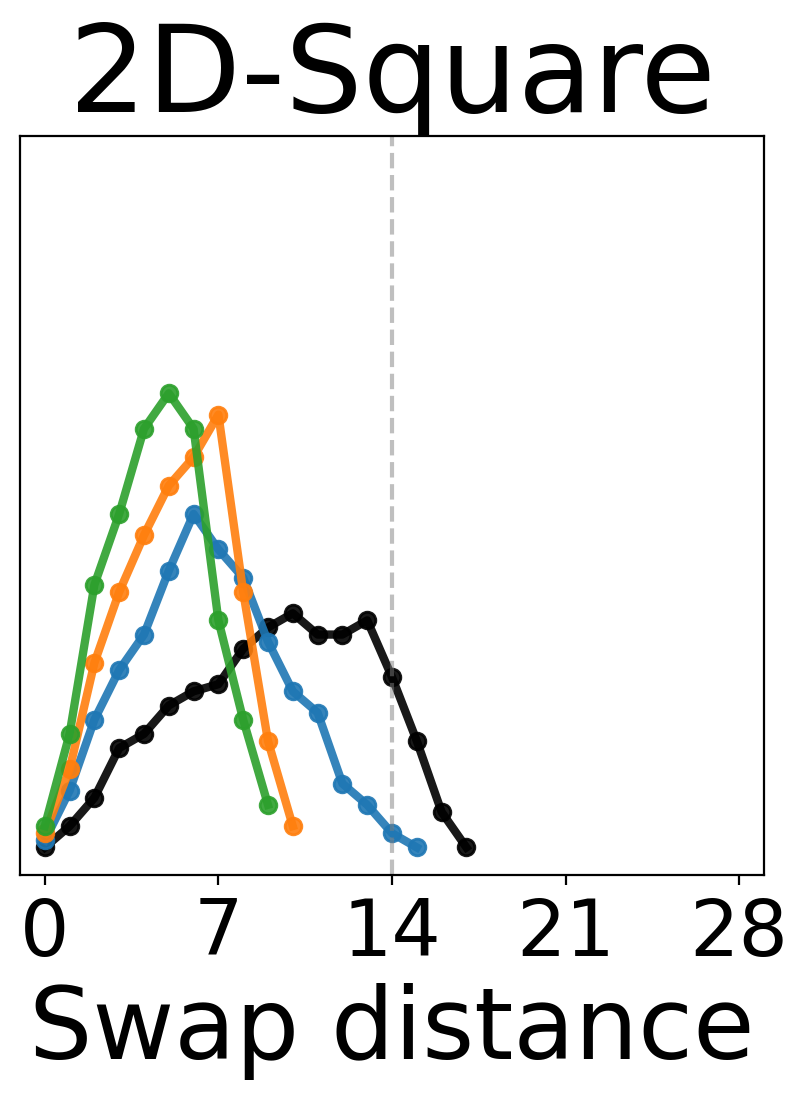}%
    \includegraphics[width=0.20592\linewidth]{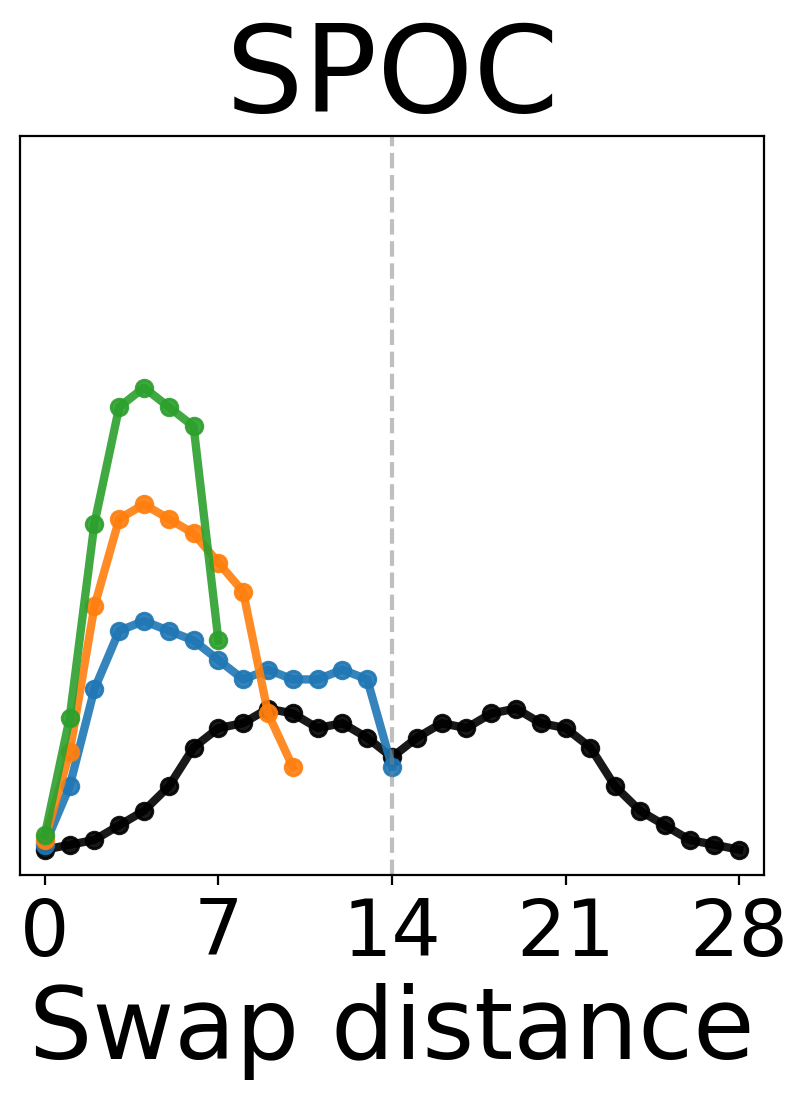}%
    \includegraphics[width=0.20592\linewidth]{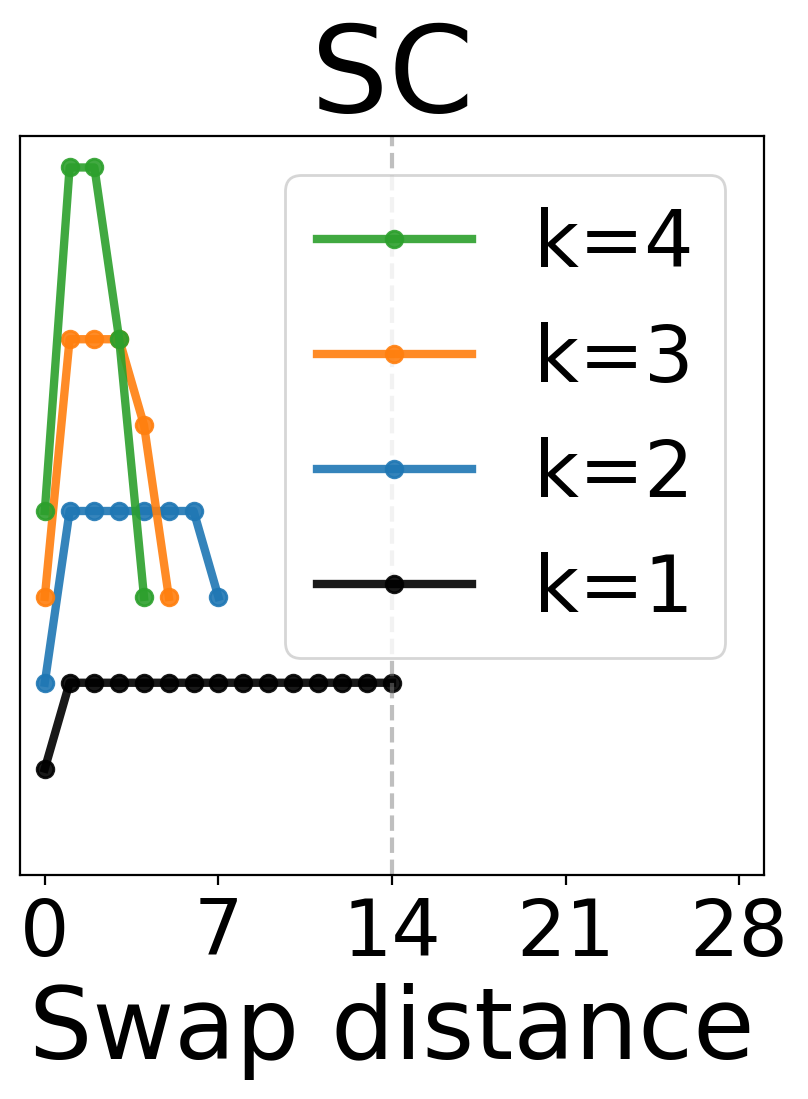}

    \caption{Histograms of the distances between the votes 
    and their closest $k$-Kemeny ranking, for $k \in \{1,2,3,4\}$.
    (For the remaining models, see the appendix.) 
    }
    \label{fig:histogram}
\end{figure}

\subsection{Diversity of the Domains}

For each of our domains---treated as an election $\calD$---we have
computed vector $\kappa(\calD)$ and present it on
\Cref{fig:diversity}. We repeated the computation 100 times and show
averaged results, together with standard deviation in the shaded areas
(nearly invisible in most cases; the randomness is over both the exact
choice of the domain---for SC and the Euclidean ones---and on the
randomness within our $k$-Kemeny heuristic). The plot also includes IC
for reference (in this case, for each of the 100 repetitions we
sampled new 512 votes from IC). As argued, for two domains $\calD'$
and $\calD''$ where $\kappa(\calD')$ dominates $\kappa(\calD'')$, we
say that $\calD'$ is (certainly) more diverse
($\calD' \pref \calD''$), leading to the following ranking (we also
resolved some unclear cases and left others as ties):
\begin{align*}
  \text{GS/cat}  \pref \text{3D-Cube} \pref \{ \text{2D-Square, SPOC}\}  
                 \pref \{\text{SP/DF, GS/bal}\} \pref \text{SP} \pref \{\text{SC, 1D-Interval}\}.
\end{align*}
The three ties that we have left seem to be genuine---looking at the
$k$-Kemeny scores of these pairs of domains we do not see strong
arguments to view either of them as more diverse than the other. This
is most clearly seen in case of SC and 1D-Interval, as their
$k$-Kemeny scores are identical.

\begin{remark}
  It is well-known that every 1D-Interval election is both SP and SC,
  but there are elections that are SP and SC, which are not
  1D-Interval~\citep{che-pru-woe:j:1d-characterization,elk-fal-sko:j:spsc}.
  Yet, the set of the maximal-sized 1D-Interval domains is a subset of
  the set of maximal-sized SC domains.
\end{remark}

Our diversity ranking is in agreement with many intuitions that one
might get by looking at the microscopes in \Cref{fig:microscope_with_ic}, but
the fact that GS/cat is the most diverse among our domains, or that
SPOC and GS/bal are ranked fairly highly, is surprising. Yet, this
seems justified. Foremost, the microscopes only give approximate views
of the domains and one of their features is that the points (i.e., the
votes) that are presented on the outer part are typically much farther
apart from each other than their Euclidean distance would
suggest. Hence, e.g., the votes in the 2D-Square domain---which appear
as a nice cluster in the inner part of the plot---are closer to each
other than the SPOC votes, which are placed on the outer part.
This is clearly seen in \Cref{fig:histogram}, where we show histograms
of the distances of the votes from their closest $k$-Kemeny ranking,
for $k \in \{1,2,3,4\}$.  In particular, for each $1$-Kemeny ranking
there are some GS/cat, SPOC, and GS/bal votes that are at the maximum
possible distance from them (which is 28 for $m=8$), meaning that
these domains include reverses of their $1$-Kemeny rankings (see also \Cref{sec:reverse-symmetric} below). Generally,
GS/cat votes are typically farther from their $k$-Kemeny rankings than
the 2D-Square ones. 

\subsubsection{Diversity of Conitzer and Walsh Elections.}

In \Cref{fig:sp_diversity} we show analogous plots as in
\Cref{fig:diversity}, but for elections sampled using the Walsh and
Conitzer models.
Intuitively, the Walsh model should produce more diverse elections as
their votes cover SP more evenly, but this conclusion is not as strong
as one might expect: Indeed $1$-Kemeny score is higher for elections
generated under the Conitzer model. However, if we disregard this one
entry, then Walsh SP elections are indeed more diverse.  Conitzer
elections, on the other hand, are more polarized as their
$\kappa_E(1) - \kappa_E(2)$ value, which measures polarization, is
larger~\citep{fal-kac-sor-szu-was:c:div-agr-pol-map}.

\begin{figure}[t]
    \centering
     \begin{minipage}[t]{0.33\columnwidth}
        \centering
        \includegraphics[width=0.99\linewidth]{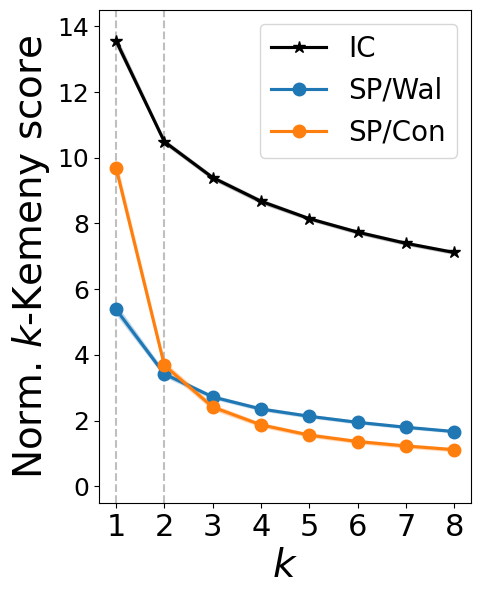}
        \captionsetup{width=0.9\linewidth, font=small}
        \caption{Average $\kappa(E)$ for SP elections sampled from Wal./Con. models.}
        \label{fig:sp_diversity}
    \end{minipage}%
    \begin{minipage}[t]{0.33\columnwidth}
        \centering
        \includegraphics[width=0.99\linewidth]{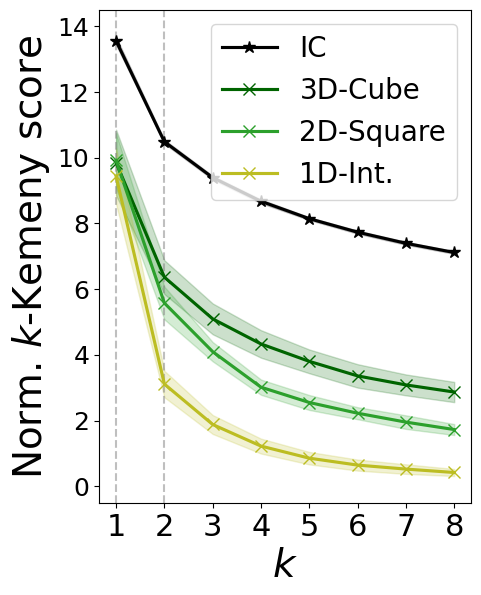} 
        \captionsetup{width=0.9\linewidth, font=small}
        \caption{Average $\kappa(E)$ for Euclidean elections sampled using 1-Box.}      
        \label{fig:euc_diversity}
    \end{minipage}%
    \begin{minipage}[t]{0.33\columnwidth}
        \centering
        \includegraphics[width=0.99\linewidth]{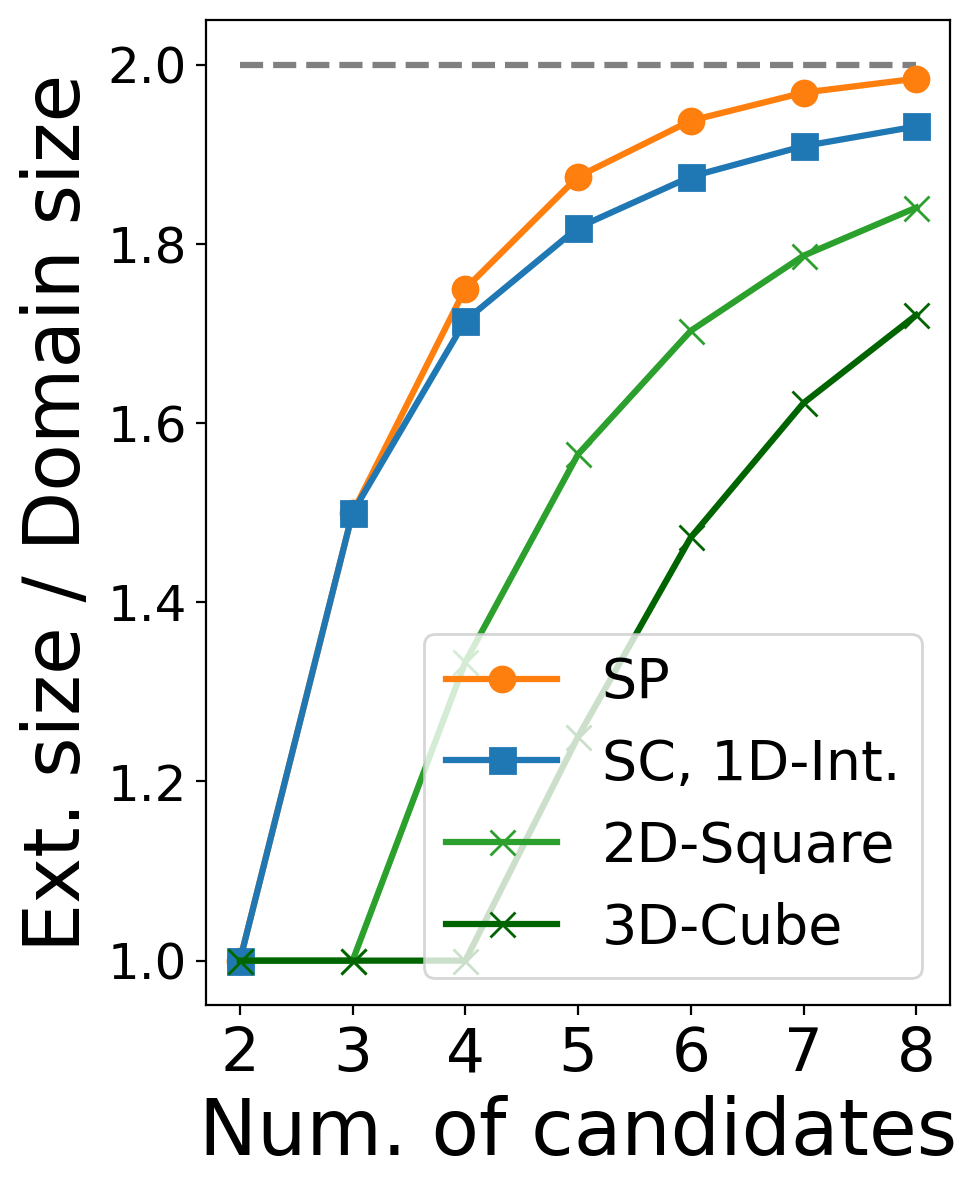}
        \captionsetup{width=0.9\linewidth, font=small}
        \caption{Sizes~of the rev.-sym. ext. of the domains over their
          original sizes.}
        \label{fig:ext_ratio}
    \end{minipage}
\end{figure}

It is tempting to
think that sampling
votes from a domain uniformly at random maximizes the diversity
of the resulting elections, but this is not always true.
For an example of a domain where this does not hold, see~\Cref{apdx:sec:sc_with_gaps_diversity}.

\subsection{Diversity of Euclidean Elections}
Euclidean domains (particularly 1D and 2D ones) are among the most frequently used structured domains in experiments~\citep{boe-fal-jan-kac-lis-pie-rey-sto-szu-was:c:guide}. While both voters' and candidates' ideal points can be drawn from various distributions, a common approach is to sample all points uniformly at random from the $1$-Box models (or isomorphic ones).

In~\Cref{fig:euc_diversity} we show an analogous plot as in
\Cref{fig:diversity}, but for elections sampled from 1D-Interval,
2D-Square, and 3D-Cube using the $1$-Box model, with IC included as a
reference point. 
As expected, the lower the dimension, the less diverse are the sampled
elections. Notably, $1$-Kemeny scores are nearly identical across all
three Euclidean models.

Yet, we observe that Euclidean elections sampled from the $1$-Box
model are visibly less diverse than their respective domains. Let us
explain this on the example of the 2D-Square domain.
Here, each possible vote corresponds to a polygon, but given the candidates' points, some polygons may lie outside of the $[-1,1]^2$
square, or they may be so small that we never ``hit'' them during
sampling. To measure how these two factors reduce the number of
distinct votes in our samples, we conducted the following
experiment. For each value of $r \in \{0.5,0.75,\dots,4\}$, we sampled
100 elections from the $r$-Box model. Then, using the candidates'
positions in each election, we calculated both the maximum number of
distinct votes possible within the given $r$-Box and the total number
of possible votes overall.\footnote{For experiments in this section, for a given dimension we
  sample 10 times more votes than in the domain (starting from 
  $290$ for 1D, up to $222120$ for 5D).}
Results are presented in~\Cref{fig:euclidean_box}~(left). As expected,
the maximum number of votes increases as we enlarge the box
size. However, interestingly, the number of distinct sampled votes
does not. Indeed, we find that the $r$-Box model produces most
distinct sampled votes for $r=1$ (we verified that this also gives
the highest diversity). Jointly with the fact that this model captures
the scenario where candidates and voters come from the same
population, we see it as an argument for using the $1$-Box model in
experiments, as already done. However, we also encourage the use of
impartial culture over Euclidean domains, which so far does not seem
to be done at all, as these cultures produce much more diverse
elections.

In~\Cref{fig:euclidean_box}, we also show the results of two similar
experiments.
In the first one,
we fix $r=1$
and vary
the number of candidates in 2 dimensions (center plot). 
In the second, we fix 8 candidates, but vary the number of dimensions (right plot).
In both cases, the number of distinct sampled votes is significantly smaller than the domain size.

\begin{figure}[t]
    \centering  
    \includegraphics[width=0.3541\linewidth]{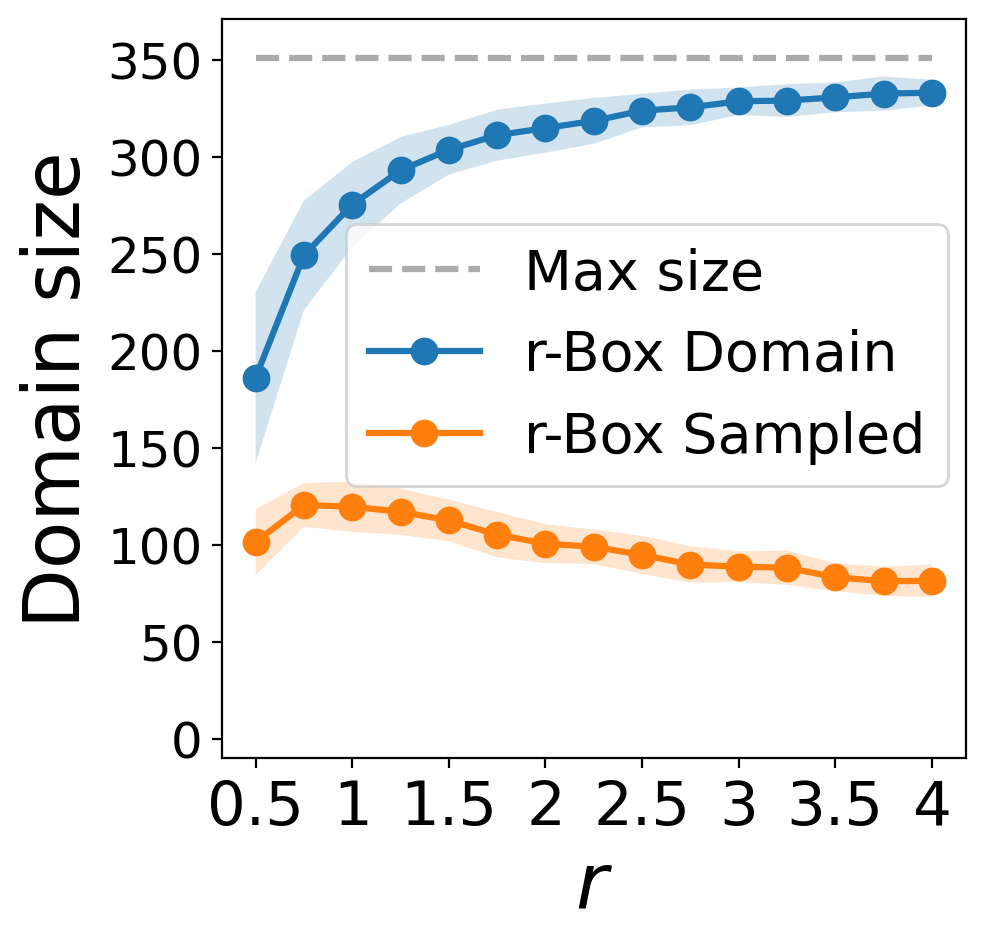}%
    \includegraphics[width=0.3228\linewidth]{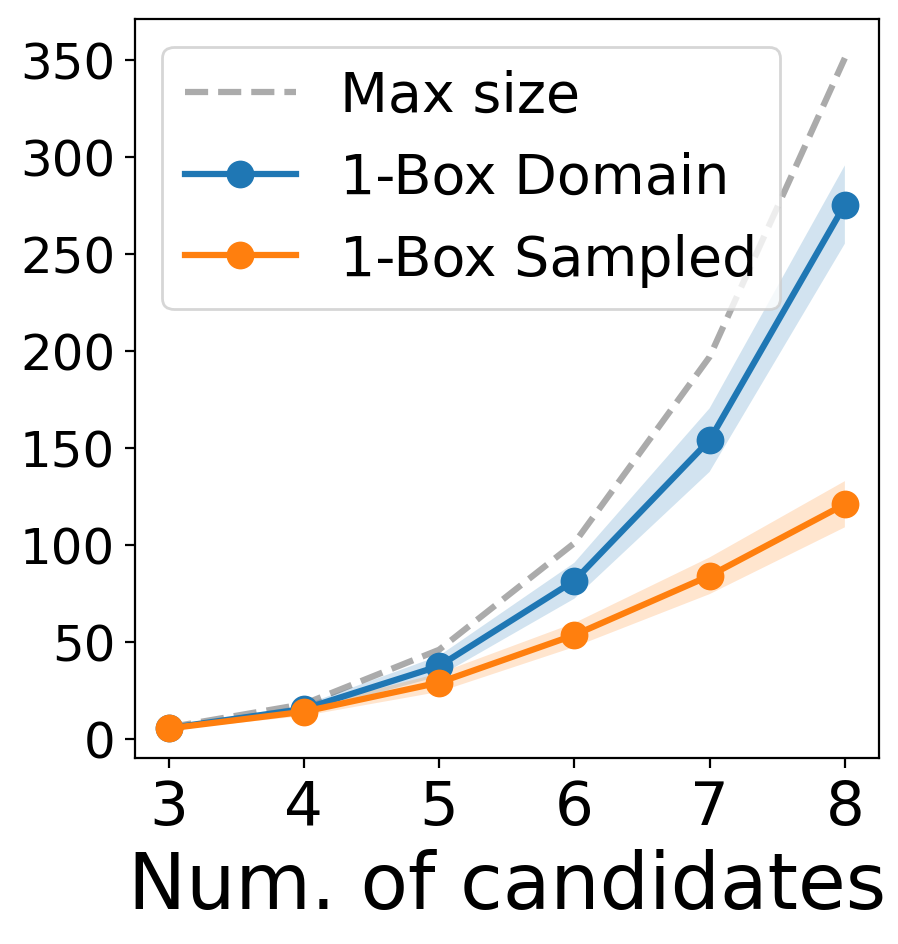}%
    \includegraphics[width=0.3228\linewidth]{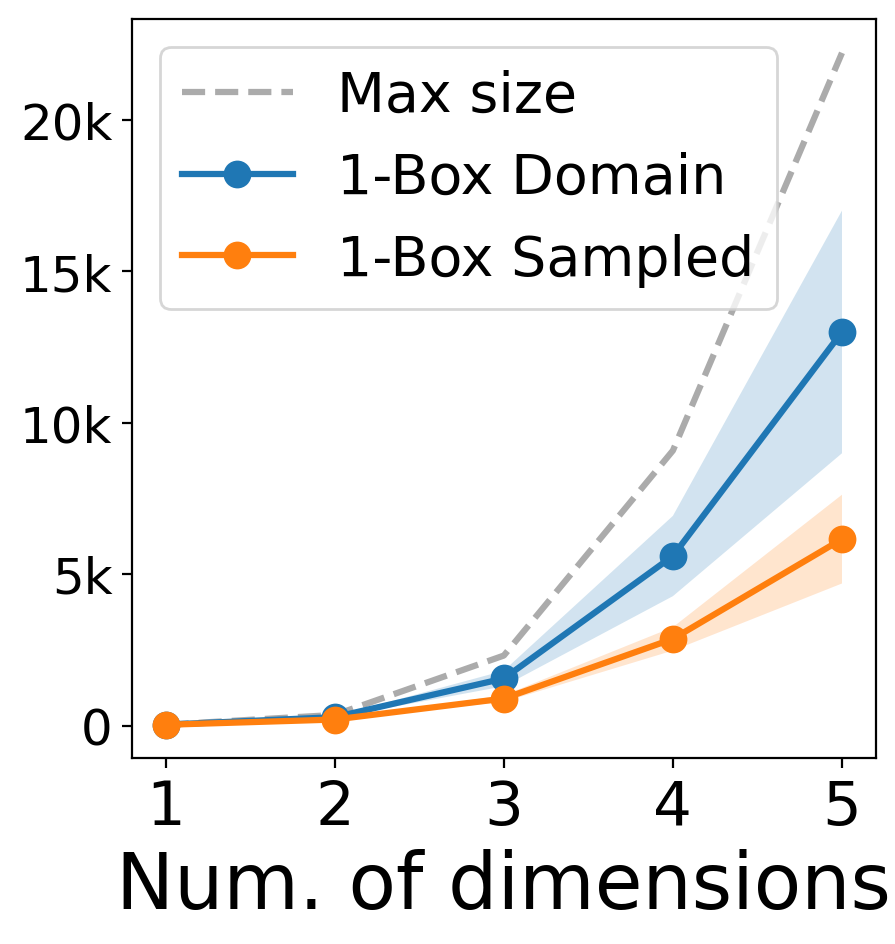}%
    \caption{Changes in the sizes of Euclidean domains (the plots on the left and in the center plot regard the 2D-Square domains; for the plot on the right, the dimension of the considered domain is on the $x$ axis).}
    \label{fig:euclidean_box}
\end{figure}

\subsection{Reverse-Symmetric Domains}\label{sec:reverse-symmetric}
We call a domain $\calD$ \emph{reverse-symmetric} if for each vote
$v~\in~\calD$, its reversal is also in $\calD$.
Similarly, it is \emph{reverse-free} if for each of its votes, its reversal
is not in $\calD$.  Among the considered domains, only GS and SPOC are
reverse-symmetric, while the SP/DF is reverse-free (in fact, all SP
on a tree domains are reverse-free, except for the SP on path).

We can extend each domain by adding a reversal of each vote (unless it is already present).
By definition, for a reverse-symmetric domain, the size of its extension equals the size of the original domain itself. For a reverse-free domain, the size of its extension is twice that of the original one.
For domains that are neither reverse-symmetric nor reverse-free, we
computed the ratio between the sizes of their extensions and
the original ones. The results in~\Cref{fig:ext_ratio} show that all
curves appear to converge towards $2$, with SP converging the fastest
and 3D the slowest.  Hence, we can classify our domains as either
being reverse-symmetric or being (nearly) reverse-free.  We encourage
using domains of both types in experiments.

\section{Conclusions}

The most important take-home messages from our work are that
caterpillar group-separable elections are much more diverse than one
might think, and that sampling votes from Euclidean domains uniformly
at random may lead to very different elections than using the natural
approach based on sampling voter points. Consequently, we encourage
the use of these means of sampling preference data in experiments.

\section*{Acknowledgements}

Tomasz Wąs was supported by UK Engineering and Physical Sciences Research Council (EPSRC) under grant EP/X038548/1.
This project
has received funding from the European Research Council (ERC) under
the European Union’s Horizon 2020 research and innovation programme
(grant agreement No 101002854).
\begin{center}
  \includegraphics[width=3.5cm]{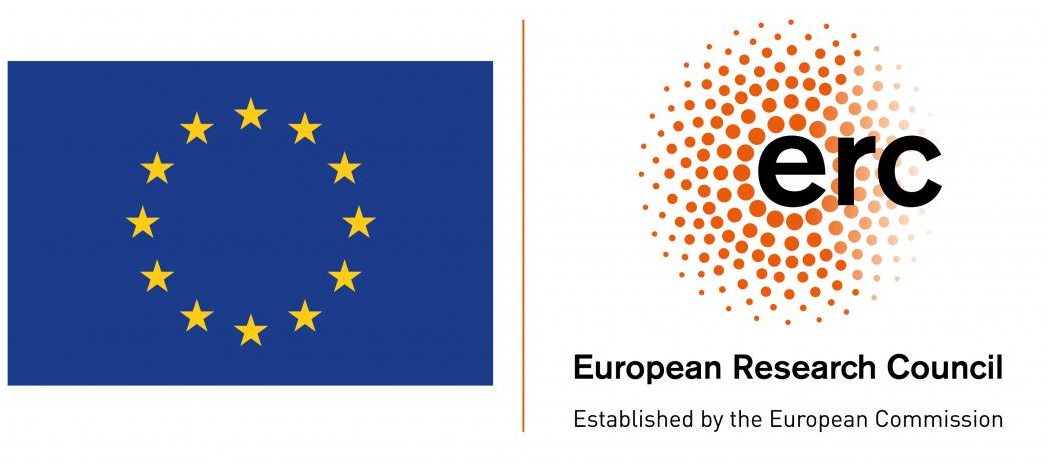}
\end{center}

\bibliography{bib-all}

\clearpage

\appendix

\section{Missing Proofs}\label{app:proofs}

In this appendix,
we provide the missing proofs that were omitted in the
main part of the paper.

The next proposition is a well-known property of Kemeny rankings that
will be crucial for many of our theoretical results. It has been
expressed in various forms and shapes by a number of different authors
(see, e.g., the works of \citet{bar:j:kemeny-condorcet},
\citet{tru:t:kemeny-condorcet}, and
\citet{bre-geo-isr-kel:c:sp-opinion-updates-kemeny}).

\begin{proposition}\label{pro:kemeny-condorcet}
  Let $E = (C,V)$ be an election and let $A, B$ be two disjoint
  subsets of candidates such that $C = A \cup B$. If for each two
  candidates $a \in A$, $b \in B$ we have that at least half of the
  voters prefer $a$ to $b$, then there is a Kemeny ranking $r$ for $E$
  such that $r \colon A \pref B$. If for each $a \in A$ and each
  $b \in B$ a strict majority of voters prefers $a$ to $b$ then
  for every Kemeny ranking $r$ we have $r \colon A \pref B$.
\end{proposition}
This implies that if an election has a Condorcet ranking then it also
is a Kemeny ranking.
\begin{corollary}\label{cor:kemeny-condorcet}
  Consider an election $E = (C,V)$ where the votes belong to some
  Condorcet domain. Each Condorcet ranking for $E$ is also a Kemeny
  ranking for $E$.
\end{corollary}

\appendixProofs

\section{Domain Sizes}
\label{app:domain-sizes}

In this appendix, we provide the maximum sizes
of some of the considered domains.
Assume that there are $m \ge 2$ candidates.
\begin{itemize}
    \item
    The single-peaked domain
    and group-separable domains:
    $2^{m-1}$
    (\href{https://oeis.org/A131577}{OEIS: A131577}).
    \item
    The single-peaked double-forked domain:
    $16 \cdot (2^{m-3} - 1)$,
    for $m \ge 5$
    (\href{https://oeis.org/A175164}{OEIS: A175164}).
    \item
    The SPOC domain:
    $m \cdot 2^{m-2}$
    (\href{https://oeis.org/A057711}{OEIS: A057711}).
    \item
    Single-crossing domains and 1D Euclidean domains:
    $\binom{m}{2}+1$
    (\href{https://oeis.org/A000124}{OEIS: A000124}).
    \item
    2D Euclidean domains:
    $s(m,m) + s(m,m-1) + s(m,m-2)$,
    where $s(n,k)$ are the unsigned
    Stirling numbers of the first
    kind~\citep{goo-tid:j:euclidean-preferences}
    (\href{https://oeis.org/A308305}{OEIS: A308305}).
\end{itemize}

\section[Additional Experiments]{Additional Experiments \& Missing Plots}
In this appendix,
we provide additional experimental results
that were omitted in the
main part of the paper.

\begin{figure}[t]
    \centering
    \begin{minipage}{0.6\textwidth}
        \centering
        \includegraphics[width=0.95\linewidth]{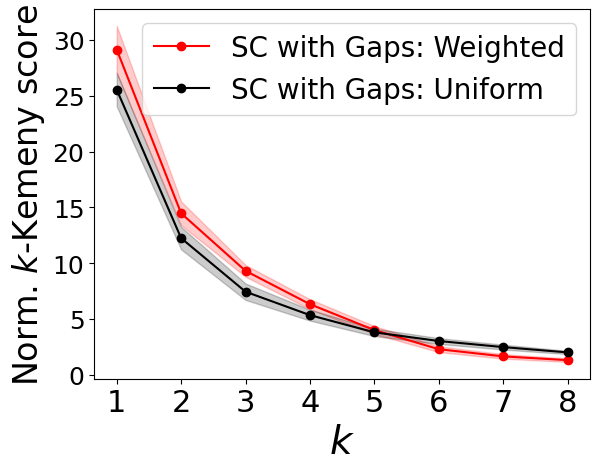}
        \caption{Diversity of sampled data.}
        \label{apdx:fig:sc_with_gaps_diversity}
    \end{minipage}%
    \begin{minipage}{0.4\textwidth}
        \centering
        \includegraphics[width=0.91\linewidth]{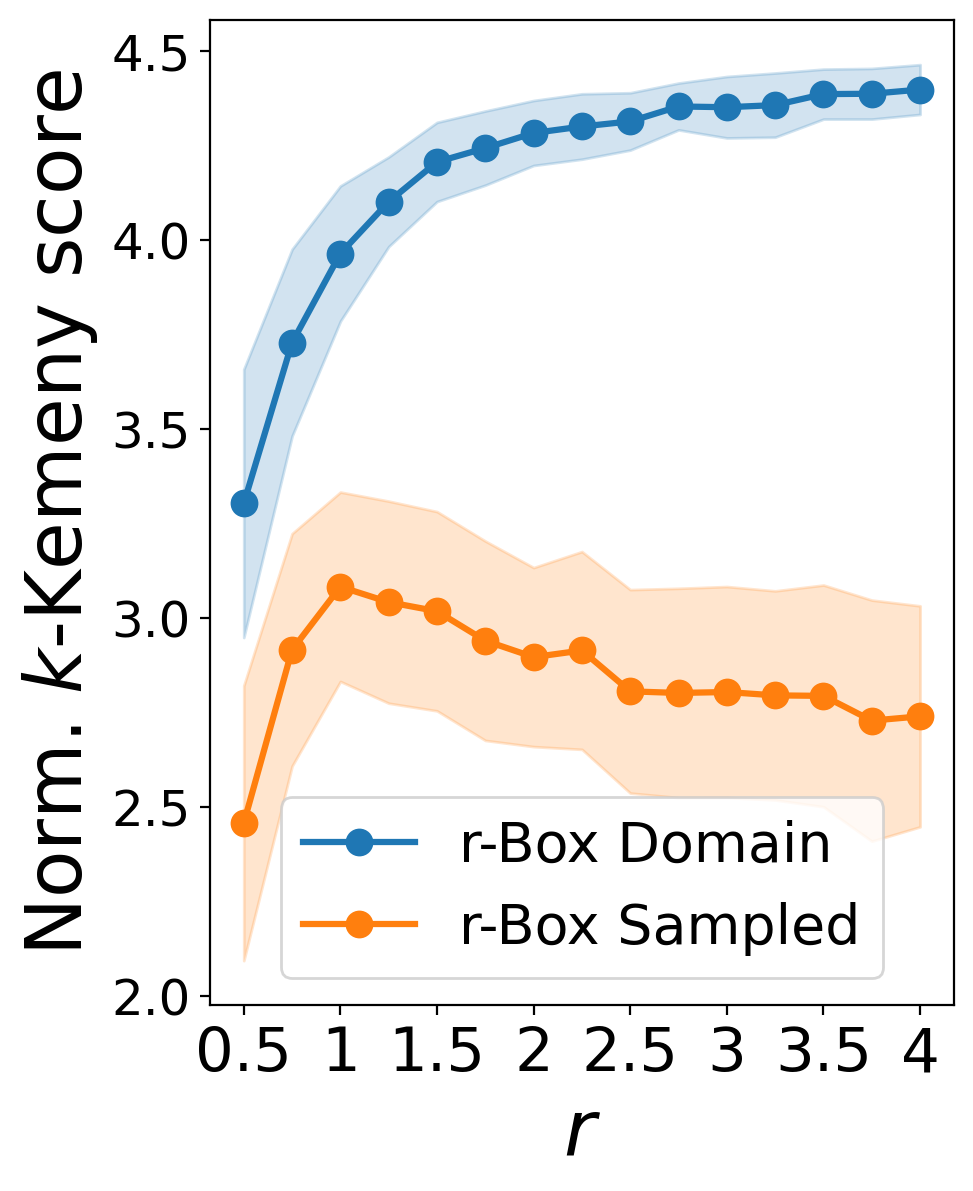}
        \caption{Diversity of sampled data.}
        \label{apdx:fig:euclidean_box}
    \end{minipage}
\end{figure}

\subsection{IC is Not Always the Most Diverse}
\label{apdx:sec:sc_with_gaps_diversity}

Consider the following domain structure: We start with a maximal SC
domain with ordered votes.  To build our domain, we take the first
vote, then omit $t$ votes, then take two votes and omit $t$ votes
again.  We continue by taking four votes and omitting $t$ votes, and
so on---always omitting $t$ votes while doubling the number of
selected votes each time.

We may either sample votes from this domain uniformly at random, or we
can use an alternative method using weighted probabilities: The first
vote included in the domain has a weight of $1$, the next two votes
have weights of $\nicefrac{1}{2}$, the next four votes have weights of
$\nicefrac{1}{4}$, and so on.  Sampling votes using these weights can
achieve higher diversity than sampling uniformly at random from the
same domain.

We consider a single-crossing domain with $16$ candidates and $t=12$
gap size.  Results are presented
in~\Cref{apdx:fig:sc_with_gaps_diversity}. The fact that the lines cross at
$k=5$ is simply because, for $16$ candidates and $t=12$ gap size,
there are $5$ groups of voters.

\subsection{Euclidean Box}

In~\Cref{apdx:fig:euclidean_box} we present a corresponding plot 
to the one presented in~\Cref{fig:euclidean_box} (left), but on the $y$-axis,
instead of showing the domain size, we present the normalized $k$-Kemeny score.

\subsection{More Microscopes}

In~\Cref{apdx:fig:microscope_without_ic} we present an analogous plot to the one
presented in~\Cref{fig:microscope_with_ic}, but without any additional IC votes.
While the microscope with IC votes shows how a given domain looks compared to
(an approximation of) the space of all possible votes, 
the one without IC votes presents the domain's internal structure better. 

In~\Cref{apdx:fig:ext_microscope} we present a microscope of extended domains, 
described in~\Cref{sec:reverse-symmetric}.

\subsection{More Histograms}
In~\Cref{apdx:fig:histogram} we present the histograms omitted in~\Cref{fig:histogram}.

\subsection{Scalability}
In~\Cref{apdx:fig:diversity_scalability} we present similar results to those in~\Cref{fig:diversity}, but for 6,8,10, and 12 candidates. In~\Cref{apdx:fig:euc_diversity_scalability}  we present similar results to those in~\Cref{fig:euc_diversity}, but for 6,8,10, and 12 candidates. In~\Cref{apdx:fig:sp_diversity_scalability}  we present similar results to those in~\Cref{fig:sp_diversity}, but for 6,8,10, and 12 candidates.

\begin{figure}[th]
    \centering
    \includegraphics[width=0.9\linewidth]{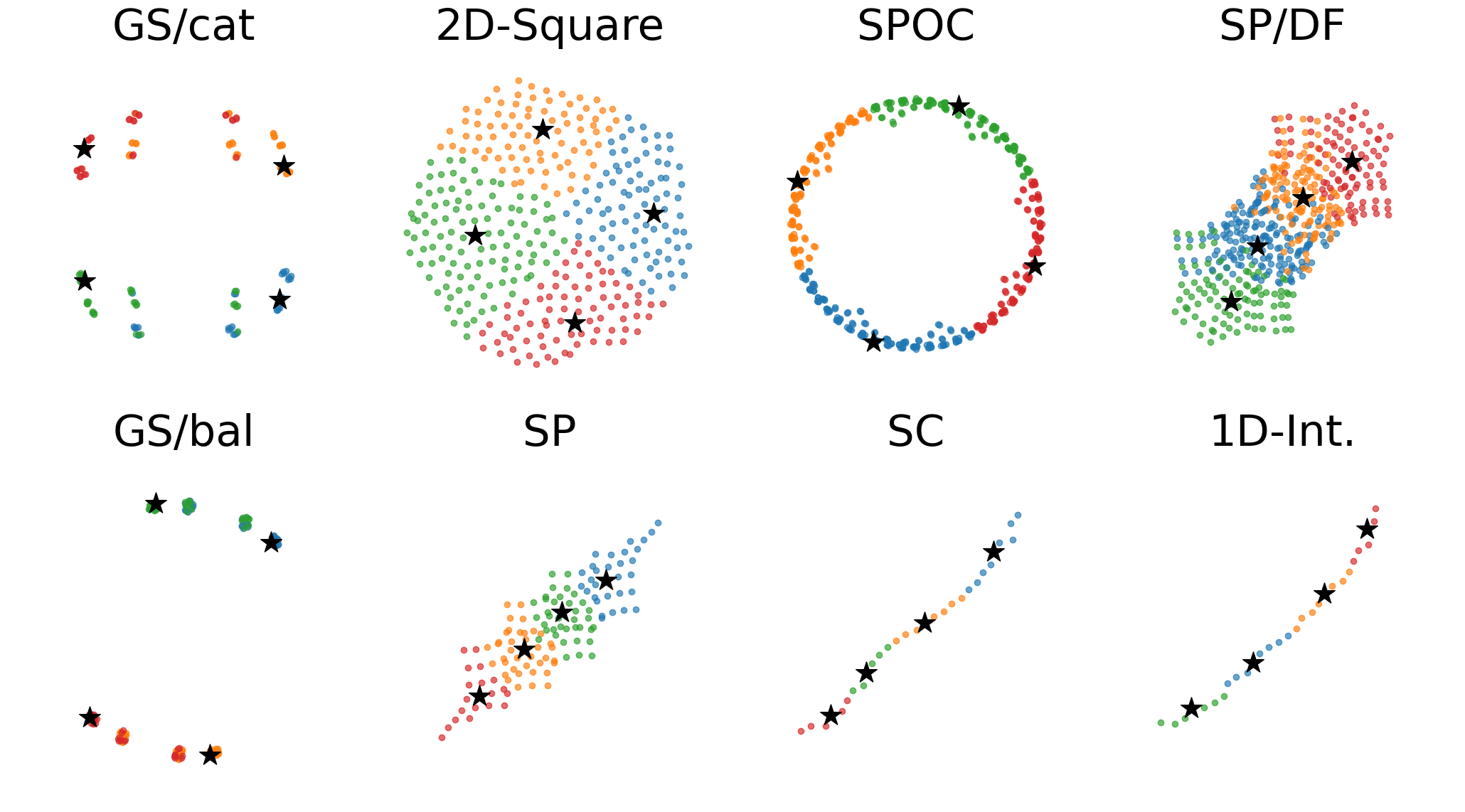}
    \caption{Microscopes of our domains without IC.}
    \label{apdx:fig:microscope_without_ic}
\end{figure}

\begin{figure}[th]
    \centering
    \includegraphics[width=0.9\linewidth]{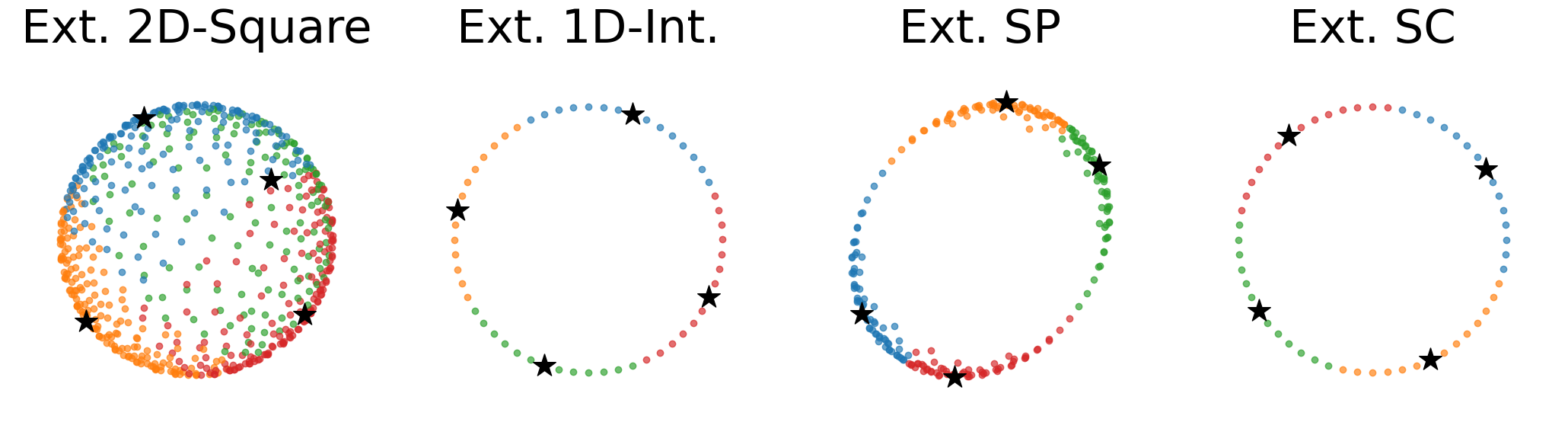}
    \caption{Extended Microscope without IC.}
    \label{apdx:fig:ext_microscope}
\end{figure}

\begin{figure}[thb]
    \centering
    \includegraphics[width=0.2826\linewidth]{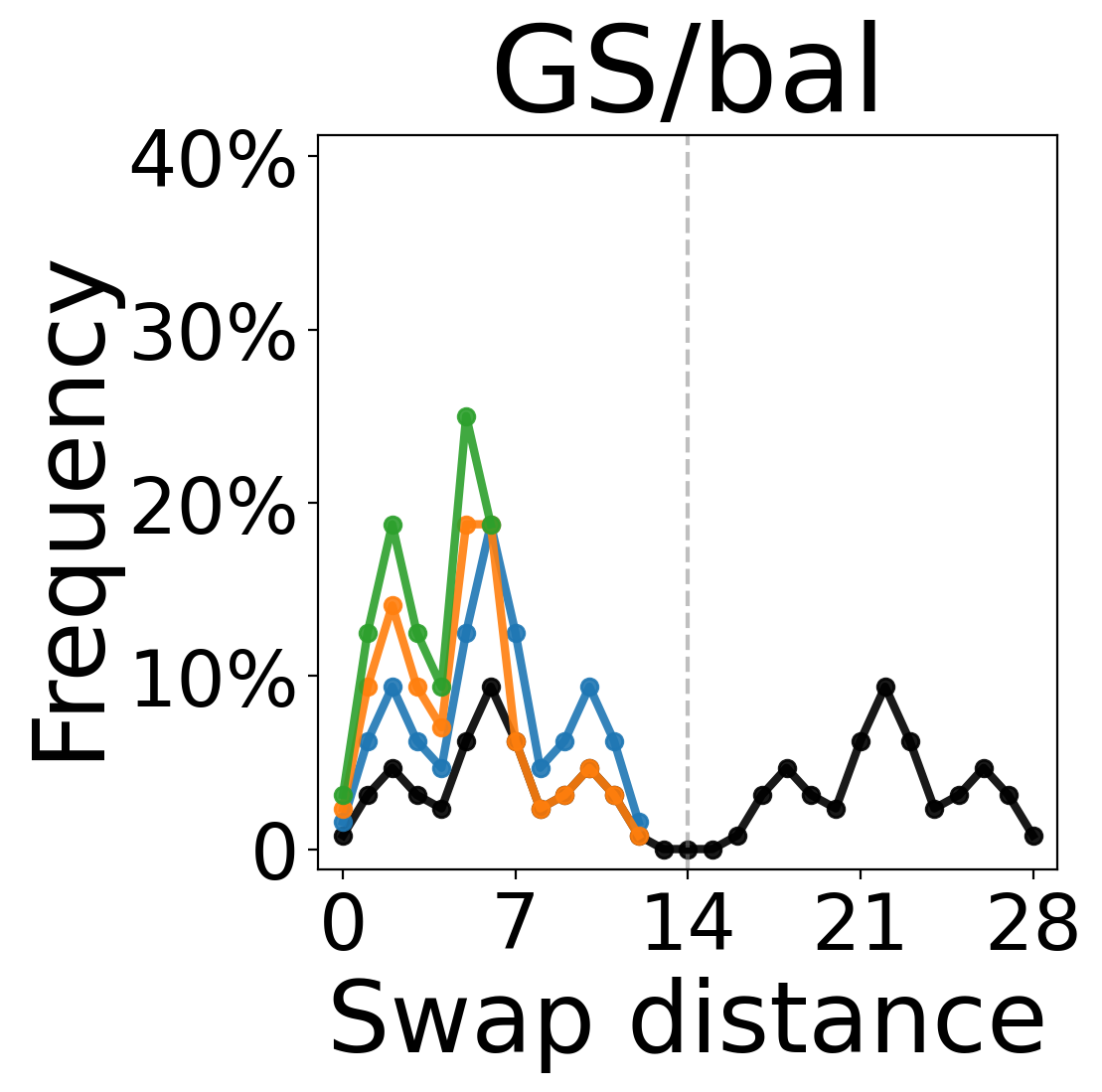}%
    \includegraphics[width=0.20592\linewidth]{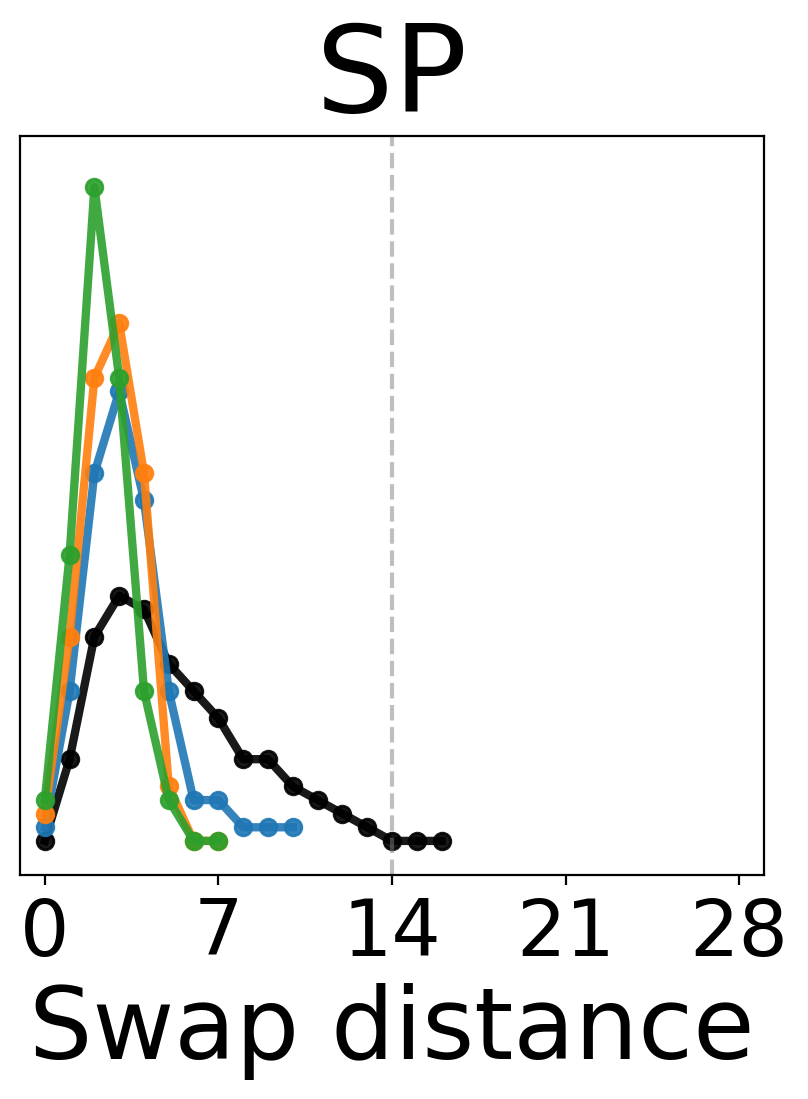}%
    \includegraphics[width=0.20592\linewidth]{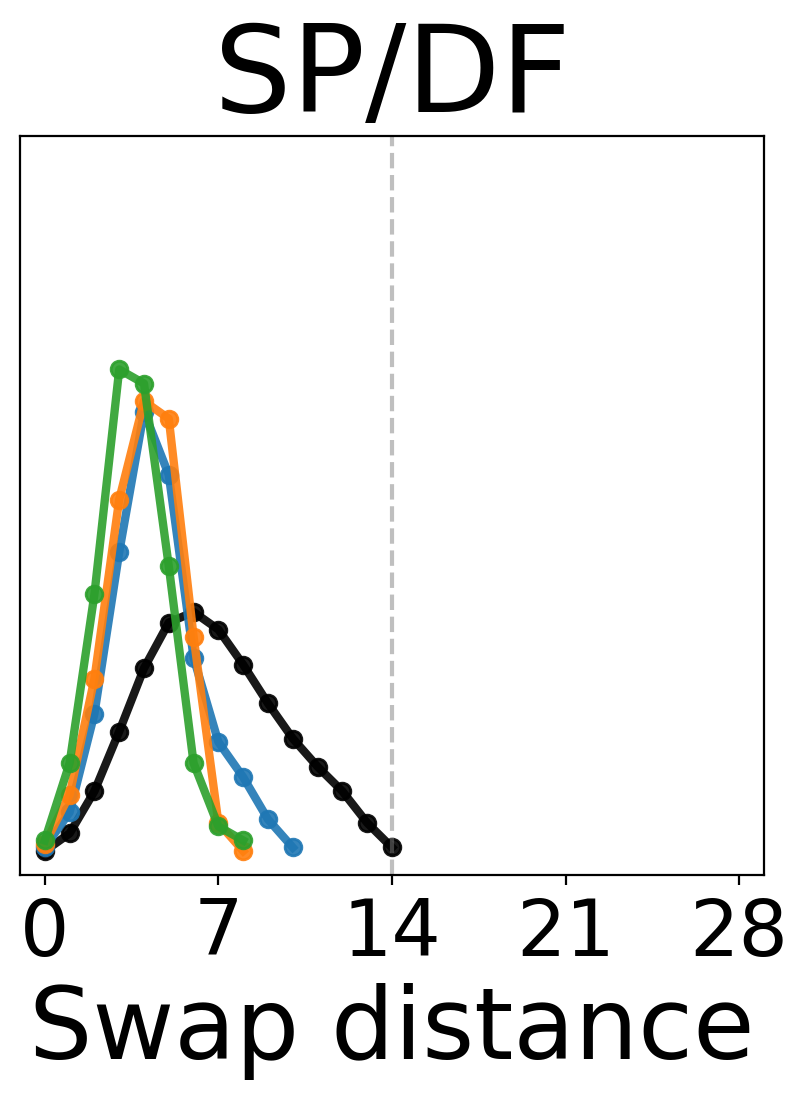}%
    \includegraphics[width=0.20592\linewidth]{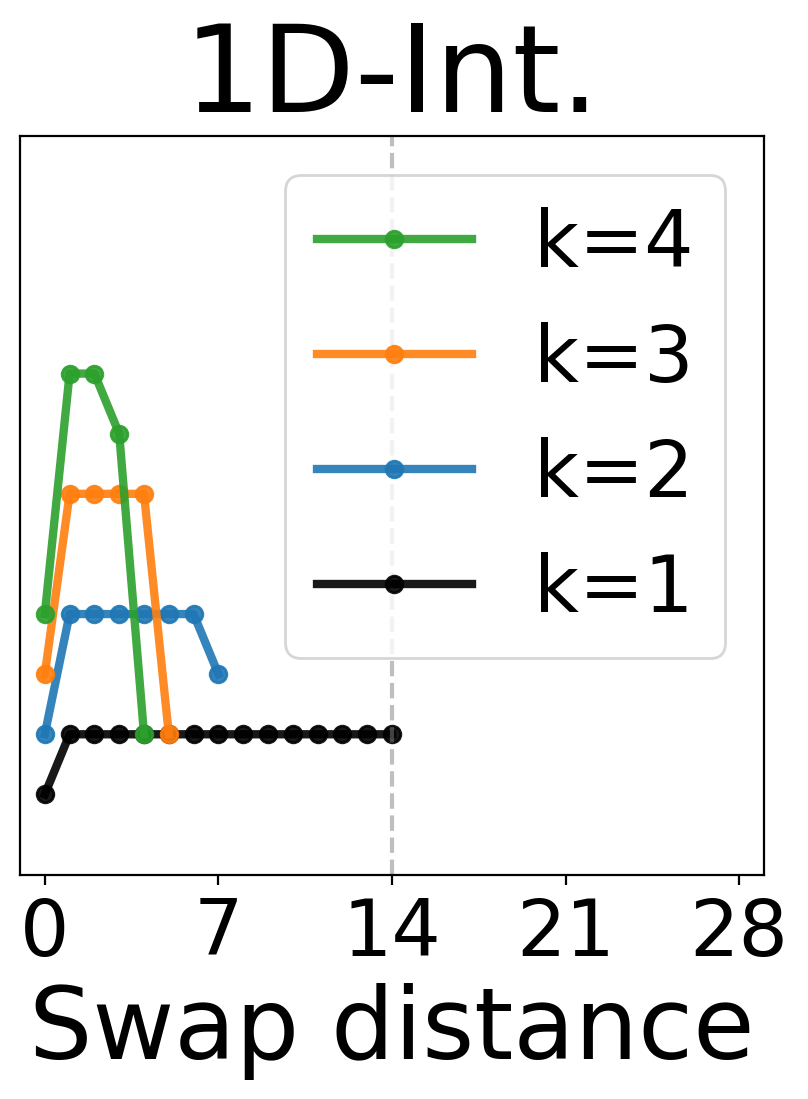}

    \caption{Histograms of the distances between the votes and their closest $k$-Kemeny ranking, for $k \in \{1,2,3,4\}$.}
    \label{apdx:fig:histogram}
\end{figure}

\begin{figure*}[th]
    \centering
    \begin{subfigure}[t]{0.25\columnwidth}
        \centering
        \includegraphics[width=1\linewidth]{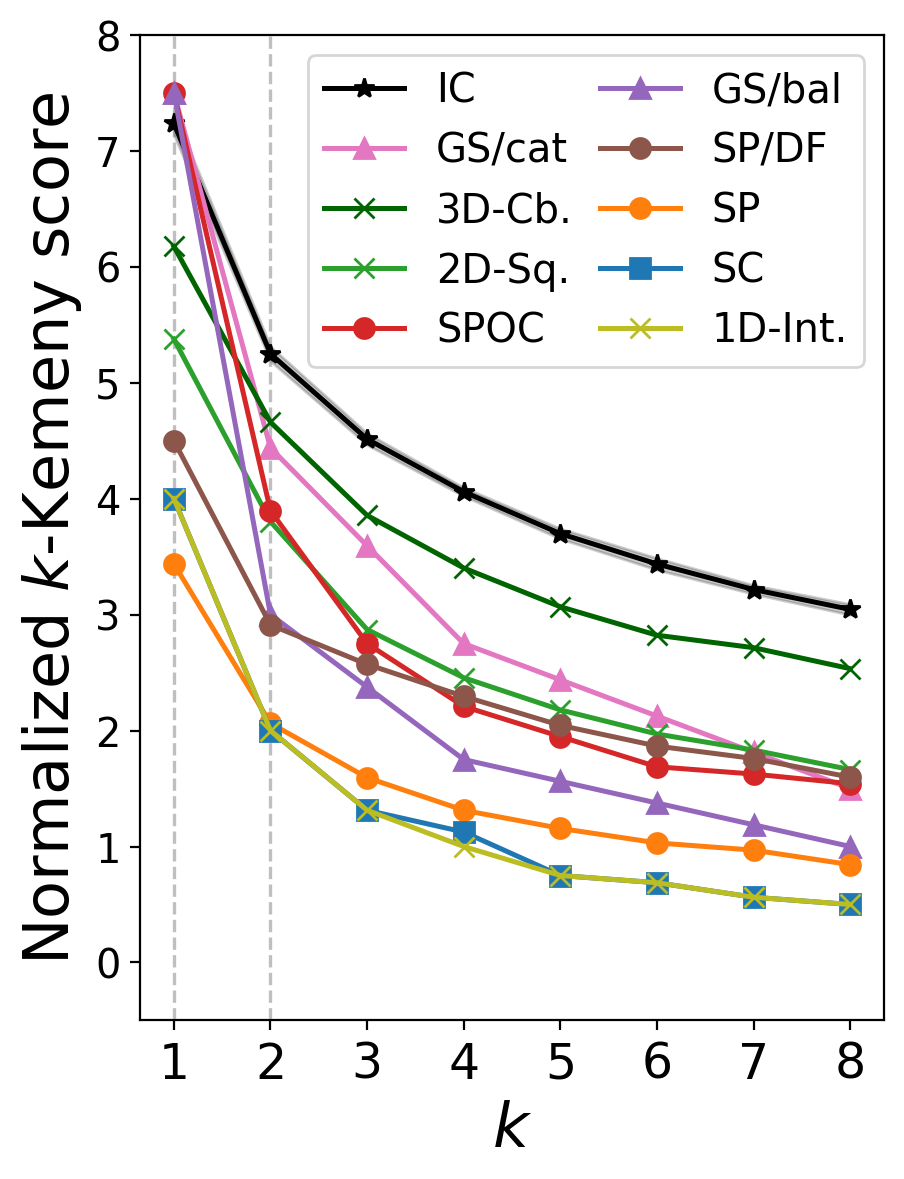}
        \caption{6 candidates.}
    \end{subfigure}%
    \begin{subfigure}[t]{0.25\columnwidth}
        \centering
        \includegraphics[width=1\linewidth]{img/diversity/domain_kemeny_m8_k8.png}
        \caption{8 candidates.}
    \end{subfigure}%
    \begin{subfigure}[t]{0.25\columnwidth}
        \centering
        \includegraphics[width=1\linewidth]{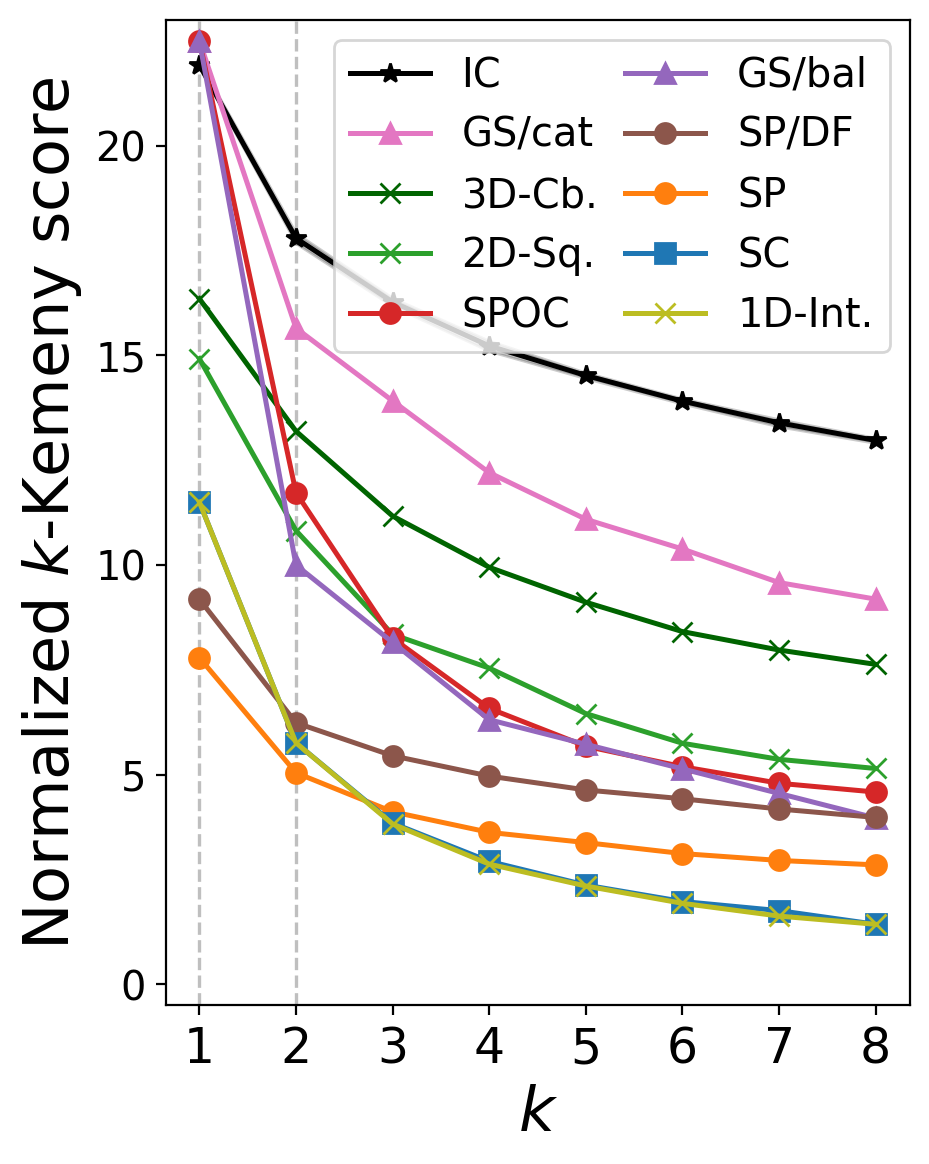}
        \caption{10 candidates.}
    \end{subfigure}%
    \begin{subfigure}[t]{0.25\columnwidth}
        \centering
        \includegraphics[width=1\linewidth]{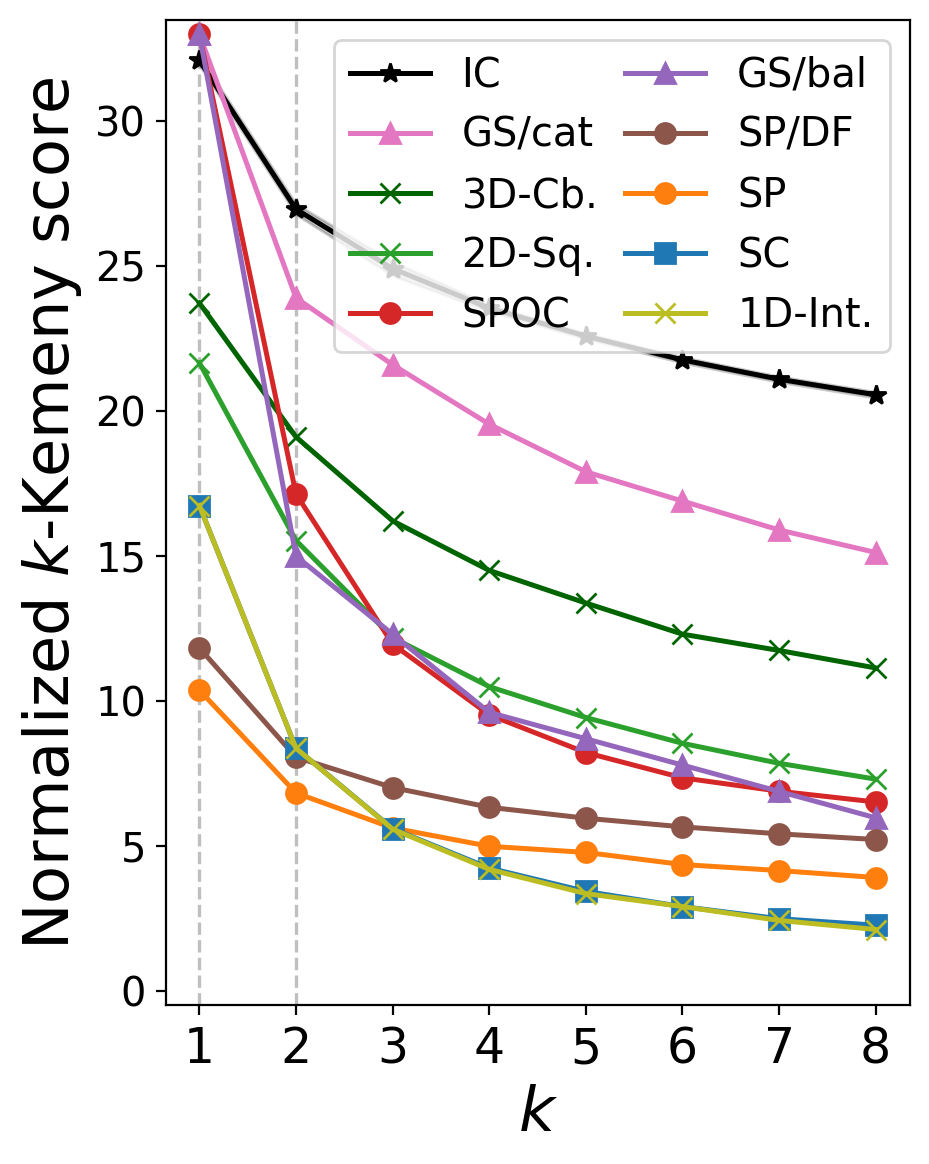}
        \caption{12 candidates.}
    \end{subfigure}
    \caption{Average $\kappa(E)$ of various domains.}
    \label{apdx:fig:diversity_scalability}
\end{figure*}

\begin{figure*}[th]
    \centering
    \begin{subfigure}[t]{0.25\columnwidth}
        \centering
        \includegraphics[width=1\linewidth]{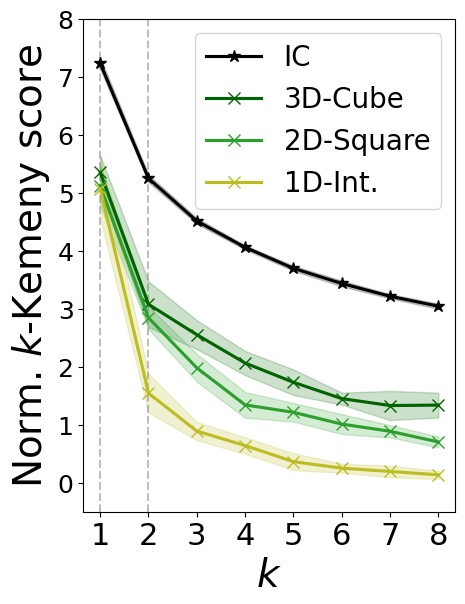}
        \caption{6 candidates.}
    \end{subfigure}%
    \begin{subfigure}[t]{0.25\columnwidth}
        \centering
        \includegraphics[width=1\linewidth]{img/diversity/sampled_kemeny_euc_diversity_m8_n512_k8.png}
        \caption{8 candidates.}
    \end{subfigure}%
    \begin{subfigure}[t]{0.25\columnwidth}
        \centering
        \includegraphics[width=1\linewidth]{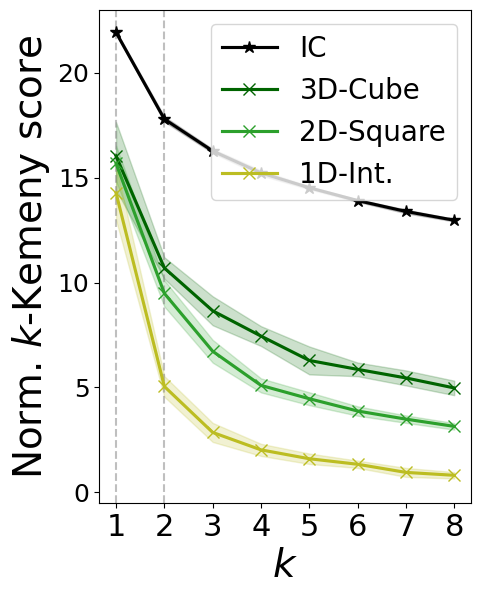}
        \caption{10 candidates.}
    \end{subfigure}%
    \begin{subfigure}[t]{0.25\columnwidth}
        \centering
        \includegraphics[width=1\linewidth]{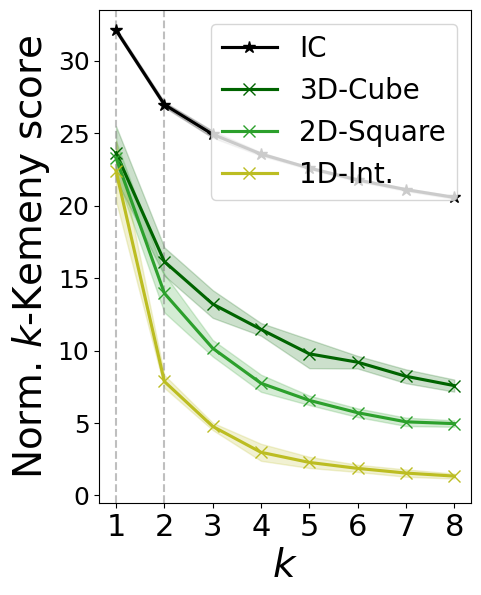}
        \caption{12 candidates.}
    \end{subfigure}
    \caption{Average $\kappa(E)$ for Euclidean elections sampled using 1-Box.}
    \label{apdx:fig:euc_diversity_scalability}
\end{figure*}

\begin{figure*}[thb]
    \centering
    \begin{subfigure}[t]{0.25\columnwidth}
        \centering
        \includegraphics[width=1\linewidth]{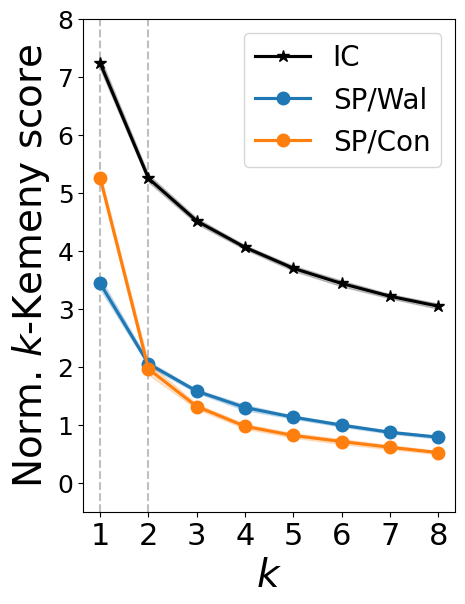}
        \caption{6 candidates.}
    \end{subfigure}%
    \begin{subfigure}[t]{0.25\columnwidth}
        \centering
        \includegraphics[width=1\linewidth]{img/diversity/sampled_kemeny_sp_diversity_m8_n512_k8.png}
        \caption{8 candidates.}
    \end{subfigure}%
    \begin{subfigure}[t]{0.25\columnwidth}
        \centering
        \includegraphics[width=1\linewidth]{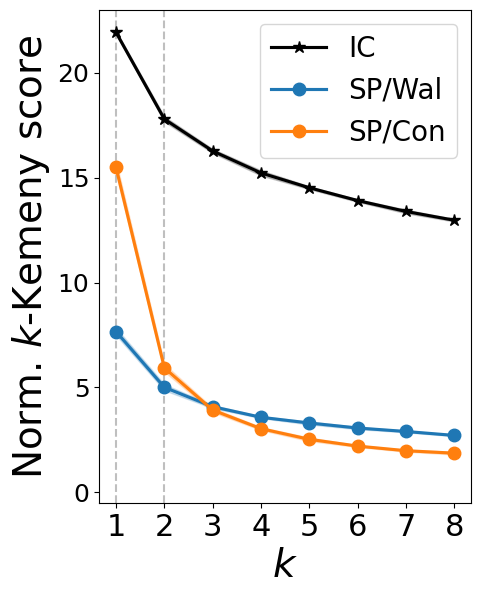}
        \caption{10 candidates.}
    \end{subfigure}%
    \begin{subfigure}[t]{0.25\columnwidth}
        \centering
        \includegraphics[width=1\linewidth]{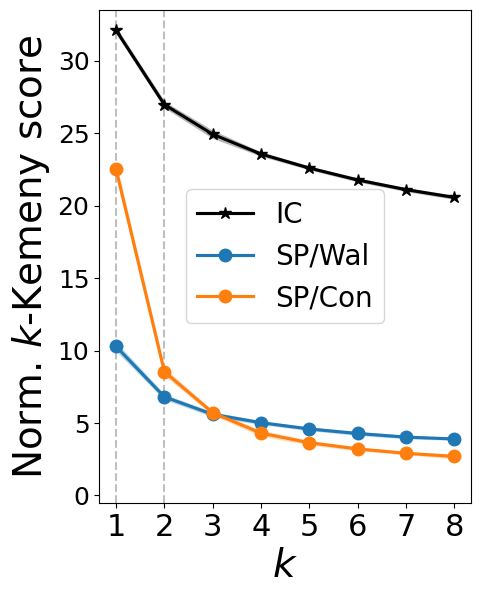}
        \caption{12 candidates.}
    \end{subfigure}
    \caption{Average $\kappa(E)$ for SP elections sampled from Wal./Con. models.}
    \label{apdx:fig:sp_diversity_scalability}
  \end{figure*}

\section{Heuristic Evaluation}
We analyze the performance of our heuristic by comparing the average scores obtained through heuristic local search against those from brute force (optimal solution) across various parameter sets. Each reported value represents the average of ten instances.

The results are presented in Tables \Cref{apdx:tab:domain1}, \Cref{apdx:tab:domain2}, \Cref{apdx:tab:domain3}, \Cref{apdx:tab:sampled1}, \Cref{apdx:tab:sampled2}, and \Cref{apdx:tab:sampled3}. As we can see, in most cases our heuristic found the optimal solution.

\begin{table*}[t]
\centering
\begin{tabular}{lcccccc}
\toprule
Domain & $m=3$ & $m=4$ & $m=5$ & $m=6$ & $m=7$ & $m=8$ \\
\midrule
1D-Int. & 1.0 & 1.0 & 1.0 & 1.0 & 1.0 & 1.0 \\
GS/cat & 1.0 & 1.0 & 1.0 & 1.0 & 1.0 & 1.0 \\
GS/bal & 1.0 & 1.0 & 1.0 & 1.0 & 1.0 & 1.0 \\
SP & 1.0 & 1.0 & 1.0 & 1.0 & 1.0 & 1.0 \\
SC & 1.0 & 1.0 & 1.0 & 1.0 & 1.0 & 1.0 \\
SPOC & 1.0 & 1.0 & 1.0 & 1.0 & 1.0 & 1.0 \\
2D-Square & 1.0 & 1.0 & 1.0 & 1.0 & 1.0 & 1.0 \\
3D-Cube & 1.0 & 1.0 & 1.0 & 1.0 & 1.0 & 1.0 \\
SP/DF & N/A & N/A & 1.0 & 1.0 & 1.0 & 1.0 \\
\bottomrule
\end{tabular}
\caption{Average score obtained by heuristic local search divided by optimal solution, for $k=1$.}
\label{apdx:tab:domain1}
\end{table*}

\begin{table*}[t]
\centering
\begin{tabular}{lcccccc}
\toprule
Domain & $m=3$ & $m=4$ & $m=5$ & $m=6$ & $m=7$ & $m=8$ \\
\midrule
1D-Int. & 1.0 & 1.0 & 1.0 & 1.0 & 1.0 & 1.0 \\
GS/cat & 1.0 & 1.0 & 1.0 & 1.0 & 1.0 & 1.0 \\
GS/bal & 1.0 & 1.0 & 1.0 & 1.0 & 1.0 & 1.0 \\
SP & 1.0 & 1.0 & 1.0 & 1.0 & 1.0 & 1.0 \\
SC & 1.0 & 1.0 & 1.0 & 1.0 & 1.0 & 1.0 \\
SPOC & 1.0 & 1.0 & 1.0 & 1.0 & 1.0 & 1.0 \\
2D-Square & 1.0 & 1.0 & 1.0 & 1.0 & 1.0009 $\pm$ 0.0027 & 1.0004 $\pm$ 0.0008 \\
3D-Cube & 1.0 & 1.0 & 1.0 & 1.0008 $\pm$ 0.0015 & 1.0019 $\pm$ 0.0027 & 1.0013 $\pm$ 0.0020 \\
SP/DF & N/A & N/A & 1.0 & 1.0 & 1.0 & 1.0 \\
\bottomrule
\end{tabular}
\caption{Average score obtained by heuristic local search divided by optimal solution, for $k=2$.}
\label{apdx:tab:domain2}
\end{table*}

\begin{table*}[thb]
\centering
\begin{tabular}{lcccccc}
\toprule
Domain & $m=3$ & $m=4$ & $m=5$ & $m=6$ & $m=7$ & $m=8$ \\
\midrule
1D-Int. & 1.0 & 1.0 & 1.0 & 1.0 & 1.0 & 1.0 \\
GS/cat & 1.0 & 1.0 & 1.0 & 1.0 & 1.0 & 1.0 \\
GS/bal & 1.0 & 1.0 & 1.0 & 1.0 & 1.0 & 1.0 \\
SP & 1.0 & 1.0 & 1.0 & 1.0 & 1.0 & 1.0 \\
SC & 1.0 & 1.0 & 1.0 & 1.0 & 1.0 & 1.0 \\
\bottomrule
\end{tabular}
\caption{Average score obtained by heuristic local search divided by optimal solution, for $k=3$.}
\label{apdx:tab:domain3}
\end{table*}

\begin{table*}[thb]
\centering
\setlength{\tabcolsep}{3pt}
\begin{tabular}{lcccccc}
\toprule
Sampled & $m=3$ & $m=4$ & $m=5$ & $m=6$ & $m=7$ & $m=8$ \\
\midrule
IC & 1.0 & 1.0 & 1.0 & 1.0004 $\pm$ 0.0010 & 1.0007 $\pm$ 0.0007 & 1.0027 $\pm$ 0.0018 \\
3D-Cube & 1.0 & 1.0 & 1.0 & 1.0003 $\pm$ 0.0006 & 1.0065 $\pm$ 0.0082 & 1.0041 $\pm$ 0.0047 \\
2D-Square & 1.0 & 1.0 & 1.0 & 1.0002 $\pm$ 0.0006 & 1.0006 $\pm$ 0.0017 & 1.0014 $\pm$ 0.0018 \\
1D-Int. & 1.0 & 1.0 & 1.0 & 1.0 & 1.0 & 1.0 \\
SP/Wal & 1.0 & 1.0 & 1.0 & 1.0 & 1.0 & 1.0 \\
SP/Con & 1.0 & 1.0 & 1.0 & 1.0 & 1.0 & 1.0 \\
\bottomrule
\end{tabular}
\caption{Average score obtained by heuristic local search divided by optimal solution, for $k=1$.}
\label{apdx:tab:sampled1}
\end{table*}

\begin{table*}[!t]
\centering
\setlength{\tabcolsep}{3pt}
\begin{tabular}{lccccc}
\toprule
Sampled & $m=3$ & $m=4$ & $m=5$ & $m=6$ & $m=7$ \\
\midrule
IC & 1.0 & 1.0011 $\pm$ 0.0034 & 1.0004 $\pm$ 0.0009 & 1.0066 $\pm$ 0.0049 & 1.0075 $\pm$ 0.0032 \\
3D-Cube & 1.0 & 1.0 & 1.0002 $\pm$ 0.0005 & 1.0 & 1.0006 $\pm$ 0.0009 \\
2D-Square & 1.0 & 1.0 & 1.0 & 1.0 & 1.0006 $\pm$ 0.0018 \\
1D-Int. & 1.0 & 1.0 & 1.0 & 1.0 & 1.0 \\
SP/Wal & 1.0 & 1.0 & 1.0 & 1.0 & 1.0 \\
SP/Con & 1.0 & 1.0 & 1.0 & 1.0 & 1.0 \\
\bottomrule
\end{tabular}
\caption{Average score obtained by heuristic local search divided by optimal solution, for $k=2$.}
\label{apdx:tab:sampled2}
\end{table*}

\begin{table*}[!t]
\centering
\begin{tabular}{lccc}
\toprule
Sampled & $m=3$ & $m=4$ & $m=5$ \\
\midrule
IC & 1.0 & 1.0016 $\pm$ 0.0035 & 1.0003 $\pm$ 0.0007 \\
3D-Cube & 1.0 & 1.0 & 1.0 \\
2D-Square & 1.0 & 1.0 & 1.0 \\
1D-Int. & 1.0 & 1.0 & 1.0010 $\pm$ 0.0030 \\
SP/Wal & 1.0 & 1.0 & 1.0028 $\pm$ 0.0085 \\
SP/Con & 1.0 & 1.0 & 1.0015 $\pm$ 0.0034 \\
\bottomrule
\end{tabular}
\caption{Average score obtained by heuristic local search divided by optimal solution, for $k=3$.}
\label{apdx:tab:sampled3}
\end{table*}

\begin{figure*}
    \vspace{10cm}
\end{figure*}

\end{document}